\documentclass[10pt, conference, letterpaper]{IEEEtran}
\IEEEoverridecommandlockouts

\usepackage{cite}
\usepackage{amsmath,amssymb,amsfonts}
\usepackage{algorithm}
\usepackage[noend]{algorithmic}
\usepackage{graphicx}
\usepackage{textcomp}
\usepackage{xcolor}
\usepackage{dsfont}
\usepackage{amsthm}
\usepackage{amsmath}
\usepackage{multirow}
\usepackage{caption}
\usepackage{subcaption}
\usepackage{mathtools}

\def\BibTeX{{\rm B\kern-.05em{\sc i\kern-.025em b}\kern-.08em
    T\kern-.1667em\lower.7ex\hbox{E}\kern-.125emX}}
    
\newtheorem{definition}{Definition}
\newtheorem{lemma}{Lemma}

\newtheorem{problem}{Problem}
\newtheorem{theorem}{Theorem}

\newcounter{relctr} 
\everydisplay\expandafter{\the\everydisplay\setcounter{relctr}{0}}

\newcommand\labrel[2]{%
  \begingroup
    \refstepcounter{relctr}%
    \stackrel{\textnormal{(\alph{relctr})}}{\mathstrut{#1}}%
    \originallabel{#2}%
  \endgroup
}
\AtBeginDocument{\let\originallabel\label}

\begin{document}

\title{ 
AdaptSLAM: Edge-Assisted Adaptive SLAM with Resource Constraints via Uncertainty Minimization}

\author{Ying Chen\IEEEauthorrefmark{1}, 
Hazer Inaltekin\IEEEauthorrefmark{2}, Maria Gorlatova\IEEEauthorrefmark{1} \\
\IEEEauthorrefmark{1}Duke University, Durham, NC, \IEEEauthorrefmark{2}Macquarie University, North Ryde, NSW, Australia\\
\IEEEauthorrefmark{1}\{ying.chen151,  maria.gorlatova\}@duke.edu, \IEEEauthorrefmark{2}hazer.inaltekin@mq.edu.au}

\maketitle

\begin{abstract} Edge computing is increasingly proposed as a solution for reducing resource consumption of mobile devices running simultaneous localization and mapping (SLAM) algorithms, with most edge-assisted SLAM systems assuming the communication resources between the mobile device and the edge server to be unlimited, or relying on heuristics to choose the information to be transmitted to the edge. This paper presents AdaptSLAM, an edge-assisted visual (V) and visual-inertial (VI) SLAM system that adapts to the available communication and computation resources, based on a theoretically grounded method we developed to select the subset of keyframes (the representative frames) for constructing the best local and global maps in the mobile device and the edge server under resource constraints. We implemented AdaptSLAM to work with the state-of-the-art open-source V- and VI-SLAM ORB-SLAM3 framework, and demonstrated that, under constrained network bandwidth, AdaptSLAM reduces the tracking error by 62\% compared to the best baseline method.

\end{abstract}


\maketitle

\section{Introduction}

Simultaneous localization and mapping (SLAM), the process of simultaneously constructing a map of the environment and tracking the mobile device's pose within it, is an essential capability for a wide range of applications, such as autonomous driving and robotic navigation~\cite{rosen2021advances}. In particular, visual (V) and visual-inertial (VI) SLAM, which use cameras either alone or in combination with inertial sensors, have demonstrated remarkable progress over the last three decades~\cite{cadena2016past}, and have become an indispensable component of emerging mobile applications such as drone-based surveillance~\cite{forster2013collaborative,williams2015scalable} 
    and markerless augmented reality~\cite{ARCore, ARkit,scargill2022here,yeh20183d}. 
    
Due to the high computational demands placed by the V- and VI-SLAM on mobile devices~\cite{xu2020edge,cao2022edge,ben2020edge,ali2022edge}, offloading parts of the workload to edge servers has recently emerged as a promising solution for lessening the loads on the mobile devices and improving the overall performance~\cite{xu2020edge,cao2022edge,deutsch2016framework,ben2020edge,ali2022edge,Karrer2018CVI,li2017corb,schmuck2019ccm,wright2020cloudslam,xu2022swarmmap}. However, such approach experiences performance degradation under resource limitations and fluctuations. The existing edge-assisted SLAM solutions either assume wireless network resources to be sufficient for unrestricted offloading, or rely on heuristics in making offloading decisions. 
    By contrast, in this paper we develop an edge computing-assisted SLAM framework, which we call AdaptSLAM, that intelligently adapts to both communication and computation resources to maintain high SLAM performance. 
    Similar to prior work~\cite{deutsch2016framework,ben2020edge,ali2022edge,Karrer2018CVI,li2017corb,schmuck2019ccm,wright2020cloudslam}, AdaptSLAM  runs a real‐time tracking module and maintains a local map on the mobile device, while offloading non‐time‐critical and computationally expensive processes (global map optimization and loop closing) to the edge server. 
    However, unlike prior designs, 
    AdaptSLAM 
    uses a \emph{theoretically grounded method} to build the local and global maps of limited size, and minimize the uncertainty of the maps, laying the foundation for the optimal adaptive offloading of SLAM tasks under the communication and computation constraints.

First, 
we develop an uncertainty quantification  model for the local and global maps in edge-assisted V-SLAM and VI-SLAM. Specifically, since these 
maps are built from the information contained in the keyframes (i.e., the most representative frames)~\cite{mur2017orb,orbslam3,vins}, the developed 
model characterizes how the keyframes and the connections between them 
contribute to the uncertainty. 
\textit{To the best of our knowledge,
this is the first uncertainty quantification model for 
V-SLAM and VI-SLAM in 
edge-assisted architectures.}

Next,  we apply the developed
uncertainty quantification model to efficiently select subsets of keyframes to build local and global maps under the constraints of limited computation and communication resources. The local and global map construction is formulated as NP-hard cardinality-constrained combinatorial optimization problems~\cite{bian2017guarantees}. 
We demonstrate that the map construction problems are `close to' submodular problems under some conditions, \textit{propose a low-complexity greedy-based algorithm to obtain near-optimal solutions}, and 
present a computation reuse method to speed up 
map construction. We implement AdaptSLAM in conjunction with the state-of-the-art open-source V- and VI-SLAM ORB-SLAM3~\cite{orbslam3} framework, and evaluate the implementation with both 
simulated and real-world communication and computation conditions. 
Under constrained 
bandwidth, AdaptSLAM reduces the tracking error by 62\%  compared with the best baseline method.

To summarize, the main contributions of this paper are: (i) the first uncertainty quantification model of maps in V- and VI-SLAM under the edge-assisted architecture, 
(ii) an analytically grounded algorithm for efficiently selecting subsets of keyframes to build local and global maps under computation and communication resource budgets, 
and (iii)  a comprehensive evaluation of AdaptSLAM on two configurations of mobile devices. 
We open-source AdaptSLAM 
via GitHub.\footnote{https://github.com/i3tyc/AdaptSLAM}

The rest of this paper is organized as follows. \S\ref{sec:RelatedWork} reviews the related work, \S\ref{sec:Preliminaries}  provides the preliminaries, \S\ref{systemmodel} and \S\ref{systemmodel2} introduce the AdaptSLAM system architecture and model,  \S\ref{sec:formulation} presents the problem formulation, 
and \S\ref{sec:ARproblem1} presents the problem solutions. We
present the evaluation in \S\ref{sec:evaluation} and conclude the paper in \S\ref{sec:conclusion}.

\section{Related Work}
\label{sec:RelatedWork}

\textbf{V- and VI-SLAM.}
Due to the affordability of cameras 
and the richness of information provided by them, V-SLAM has been widely studied in the past three decades~\cite{cadena2016past}. 
It 
can be classified into 
direct approaches (LSD-SLAM~\cite{engel2014lsd}, DSO~\cite{engel2017direct}), which operate directly on pixel intensity values, and feature-based approaches (PTAM~\cite{klein2007parallel}, ORB-SLAM2~\cite{mur2017orb}, Pair-Navi~\cite{dong2019pair}), which extract salient regions in each camera frame. We focus on feature-based approaches 
since direct approaches require 
high computing power for real-time performance~\cite{cadena2016past}. 
To provide robustness (to textureless areas,  motion blur, 
illumination changes), there is a growing trend of employing \mbox{VI-SLAM}, that assists the cameras with an inertial measurement unit (IMU)~\cite{weiss2012real,vins,orbslam3}; VI-SLAM has become the de-facto standard SLAM method for 
modern augmented reality platforms~\cite{ARCore,ARkit}. In VI-SLAM, visual information and IMU data can be loosely~\cite{weiss2012real} or tightly~\cite{vins,orbslam3} coupled.  We implement AdaptSLAM based on ORB-SLAM3~\cite{orbslam3}, a state-of-the-art open-source V- and VI-SLAM system which tightly integrates visual and IMU information.

\textbf{Edge-assisted SLAM.}
Recent studies \cite{ORBBuf,williams2015scalable, riazuelo2014c2tam,deutsch2016framework,ben2020edge,Karrer2018CVI,li2017corb,schmuck2019ccm,wright2020cloudslam,xu2020edge,xu2022swarmmap,huang2022edge} have focused on offloading 
parts of SLAM workloads from mobile devices to edge (or cloud) servers  to reduce
mobile device resource consumption. 
A standard approach is to offload computationally expensive tasks 
(global map optimization, loop closing), while exploiting 
onboard computation for running the tasks critical to the mobile device's autonomy 
(tracking, local map optimization)~\cite{deutsch2016framework,ben2020edge,Karrer2018CVI,li2017corb,schmuck2019ccm,wright2020cloudslam,xu2022swarmmap}. 
Most 
edge-assisted SLAM frameworks 
assume 
wireless network resources to be sufficient for unconstrained offloading~\cite{williams2015scalable,riazuelo2014c2tam,deutsch2016framework,Karrer2018CVI,li2017corb,schmuck2019ccm}; 
some use heuristics to choose the information to be offloaded
under communication constraints~\cite{ben2020edge,xu2020edge,ORBBuf,wright2020cloudslam,xu2022swarmmap,huang2022edge}.
 Some frameworks 
only keep the newest keyframes in the local map to combat the constrained computation resources on mobile devices~\cite{Karrer2018CVI,schmuck2019ccm}. 
 Complementing
this work, we propose a theoretical framework to characterize how  keyframes contribute to the SLAM performance, 
laying the foundation for the adaptive offloading of SLAM tasks under the communication and computation 
constraints.

\textbf{Uncertainty quantification and minimization.} Recent work~\cite{khosoussi2019reliable,carlone2018attention,chen2021anchor} has 
focused on quantifying and minimizing the pose estimate uncertainty 
in V-SLAM. Since the pose estimate accuracy is difficult to obtain 
due to the lack of ground-truth poses of mobile devices, the uncertainty can 
guide the 
decision-making in SLAM systems. 
In~\cite{khosoussi2019reliable,carlone2018attention}, it 
is used for measurement selection (selecting measurements between keyframes~\cite{khosoussi2019reliable} and selecting extracted features of keyframes~\cite{carlone2018attention}); in~\cite{chen2021anchor}, it 
is used for anchor selection (selecting keyframes to make their poses have `zero uncertainty'). Complementing this work, we 
quantify the pose estimate uncertainty of both V- and VI-SLAM under the edge-assisted architecture. After the uncertainty quantification, we study the problem of selecting a subset of keyframes to minimize the uncertainty. 
This problem is largely overlooked in the literature, but is of great importance for tackling computation and communication constraints in edge-assisted SLAM.

\section{Preliminaries}
\label{sec:Preliminaries}
\subsection{Graph Preliminaries}
A directed multigraph is defined by the tuple of sets $\mathcal{G}=(\mathcal{V},\mathcal{E}, \mathcal{C})$, where $\mathcal{V}=\{v_1, \cdots, v_{|\mathcal{V}|} \}$ is the set of nodes, $\mathcal{E}$ is the set of edges, and $\mathcal{C}$ is the set of edge categories. 
Let $e=((v_i,v_j),c) \in \mathcal{E}$ denote the edge, where the nodes $v_i,v_j \in\mathcal{V}$ are the head and tail of $e$, and $c\in \mathcal{C}$ is the category of $e$. We let $w_e$ be the weight of edge $e$. We allow multiple edges from ${v}_i$ to ${v}_j$ to exist, and denote the set of edges from ${v}_i$ to ${v}_j$ by $\mathcal{E}_{i,j}$. Note that the edges in $\mathcal{E}_{i,j}$ are differentiated from each other by their category labels. The total edge weight from nodes ${v}_i$ to ${v}_j$ is given by $w_{i,j} = \sum\limits_{e \in \mathcal{E}_{i,j}} {{w_e}}$, which is the sum of all edge weights from  ${v}_i$ to ${v}_j$. 

The weighted Laplacian matrix $\mathbf{L}$ of graph $\mathcal{G}$ is a $|\mathcal{V}| \times |\mathcal{V}|$ matrix where the $i,j$-th element $\mathbf{L}_{i,j}$ is given by:

\[{{\bf{L}}_{i,j}} = \left\{ {\begin{array}{*{20}{c}}
 - {w_{i,j},}&{i \ne j}\\
{\sum\limits_{e \in \mathcal{E}_i} {{w_e},} }&{i = j}
\end{array}} \right.,\]
where $\mathcal{E}_i \subseteq \mathcal{E}$ is the set of all edges whose head is node $v_i$.
 The reduced Laplacian matrix $\tilde{\mathbf L}$ is obtained by removing an arbitrary node (i.e., removing the row and
column associated to the node) from ${\mathbf L}$.

\subsection{Set Function} 
We define a set function $f$ for a finite set $V$ as a mapping $f:{2^V} \to \mathds{R}$ that assigns a value
$f\left(S\right)$ to each subset $S\subseteq V$. 

\textbf{Submodularity.} A set function $f$ is submodular if
$f\left( L \right) + f\left( S \right) \geqslant f\left( {L \cup S} \right) + f\left( {L \cap S} \right)$ for all $L,S \subseteq V$.

\textbf{Submodularity ratio.} The submodularity ratio of a set function $f$ with respect to a parameter $s$ is 
\begin{equation}
\label{eq:ratio}
\gamma =\mathop {\min }\limits_{L \subseteq V,S \subseteq V,\left| S \right| \leqslant s,x \in {V\setminus(S \cup L)}} \frac{{f\left( {L \cup \left\{ x \right\}} \right) - f\left( L \right)}}{{f\left( {L \cup S \cup \left\{ x \right\}} \right) - f\left( {L \cup S} \right)}}.
\end{equation}
where we define $0/0 \coloneqq 1$.

The cardinality-fixed maximization problem is 
\begin{equation}
\label{eq:SMP}
\mathop {\max }\limits_{_{S \subseteq V,\left| S \right| = s}} f\left( S \right).
\end{equation}
The keyframe selection optimization is closely related to the cardinality-fixed maximization problem introduced above, which is an NP-hard problem \cite{nemhauser1978analysis}. However, for submodular set functions, there is an efficient greedy approach that will come close to the optimum value for \eqref{eq:SMP}, with a provable optimality gap. This result is formally stated in Theorem 1. 

  \setlength{\textfloatsep}{3pt}
 \begin{algorithm}[b]
 \caption{Greedy algorithm to solve~\eqref{eq:SMP}}
 \begin{algorithmic}[1]
  \label{alg:greedy}
 \STATE $S^\#\leftarrow\emptyset$;
 \WHILE{($\left| {{S^\# }} \right| < s$)}
 \STATE ${x^\star} \leftarrow \arg \mathop {\max }\limits_{x} 
 f({S^\# } \cup \{ x\} ) - f({S^\# })$.
$S^\# \leftarrow {S^\# } \cup \{ x^\star\} $.
  \ENDWHILE
 \end{algorithmic} 
 \end{algorithm}

\begin{theorem}
\label{theorem:submax}
 \cite{nemhauser1978analysis,das2018approximate} Given a non-negative and monotonically increasing 
 set function $f$ with a submodularity ratio~$\gamma$, let $S^{\#}$ be the solution produced by the greedy algorithm (Algorithm~\ref{alg:greedy}) and $S^{\star}$ be the solution of \eqref{eq:SMP}. Then, 
 $
f\left( {{S^\# }} \right) \geqslant \left( {1 - \exp(-\gamma)} \right)f\left( {{S^\star}} \right).$
\end{theorem}

\subsection{SLAM Preliminaries}
\label{sec:SLAM preliminaries}

The components of SLAM systems include~\cite{cadena2016past,orbslam3,vins}:

\textbf{Tracking.} The tracking module detects 2D \textit{feature points} (e.g., by extracting SIFT, SURF, or ORB descriptors) in the current frame. Each feature point corresponds to a 3D \textit{map point} (e.g., a distinguishable landmark) in the environment. 
The tracking module uses these feature points to find correspondences with a previous reference frame. It also processes the IMU 
measurements. 
Based on the correspondences in feature points and the IMU measurements, it calculates the relative pose change between the selected reference frame and the current frame. The 
module also determines if this frame should be a keyframe 
based on a set of criteria such as the similarity to the previous keyframes \cite{orbslam3}.

\textbf{Local and global mapping.} 
It
finds correspondences (of 
feature points)  between the new keyframe and the other
keyframes in the map. It then performs  map optimizations, i.e.,
estimates the keyframe poses given the 
common feature points between the keyframes and the IMU measurements.
\textit{Map optimizations are computationally expensive.}  
In edge-assisted SLAM, global mapping runs 
on the server~\cite{deutsch2016framework,ben2020edge,Karrer2018CVI,li2017corb,schmuck2019ccm,wright2020cloudslam}.

\textbf{Loop closing.} By comparing the new keyframe to all previous keyframes, the module checks if the new keyframe is revisiting a place. If so (i.e., if a loop is detected), it establishes connections between the keyframe and all related previous ones, and then performs global map optimizations. Loop closing is computationally expensive and can be offloaded to the edge server in the edge-assisted SLAM~\cite{deutsch2016framework,ben2020edge,Karrer2018CVI,li2017corb,schmuck2019ccm,wright2020cloudslam}.

\section{AdaptSLAM System Architecture}
\label{systemmodel}

The design of AdaptSLAM is shown in Fig.~\ref{architecture}. 
The mobile device, equipped with a camera and an IMU, can communicate with the edge server bidirectionally. The mobile device and the edge server cooperatively run SLAM algorithms to 
estimate the mobile device's pose and a map of the environment. AdaptSLAM optimizes the SLAM performance under computation resource limits of the mobile device and communication resource limits between the mobile device and the edge server.

We split the modules between the mobile device and the edge server similar to~\cite{deutsch2016framework,ben2020edge,Karrer2018CVI,li2017corb,schmuck2019ccm,wright2020cloudslam,xu2022swarmmap}. 
The mobile device offloads loop closing and global map optimization modules to the edge server, while running real-time tracking and local mapping onboard. 
Unlike existing edge-assisted SLAM systems~\cite{deutsch2016framework,ben2020edge,Karrer2018CVI,li2017corb,schmuck2019ccm,wright2020cloudslam,xu2022swarmmap}, \emph{AdaptSLAM aims to optimally construct the local and global maps under the computation and communication resource constraints.} 
The design of AdaptSLAM is mainly focused on two added modules, local map construction and global map construction highlighted in purple in Fig.~\ref{architecture}.
In local map construction, due to the computation resource limits, 
the mobile device selects a subset of keyframes from candidate keyframes to build a local map. 
In global map construction, to adapt to the constrained wireless connection for uplink transmission, the mobile device also selects a subset of keyframes to be transmitted to the edge server to build a global map. The AdaptSLAM optimally selects the keyframes to build local and global maps, minimizing the pose estimate uncertainty under the resource constraints.

Similar to~\cite{ben2020edge}, the selected keyframes are transmitted from the mobile device to the server, and the map after the global map optimization is transmitted from the server to the mobile device. For the uplink transmission, instead of the whole keyframe, 
the 2D feature points extracted from the keyframes are sent. For the downlink communication, the poses of the keyframes obtained by the global map optimization, and the feature points of the keyframes are transmitted.

 \begin{figure}[t]
\centering
   \includegraphics[width=0.5\textwidth]{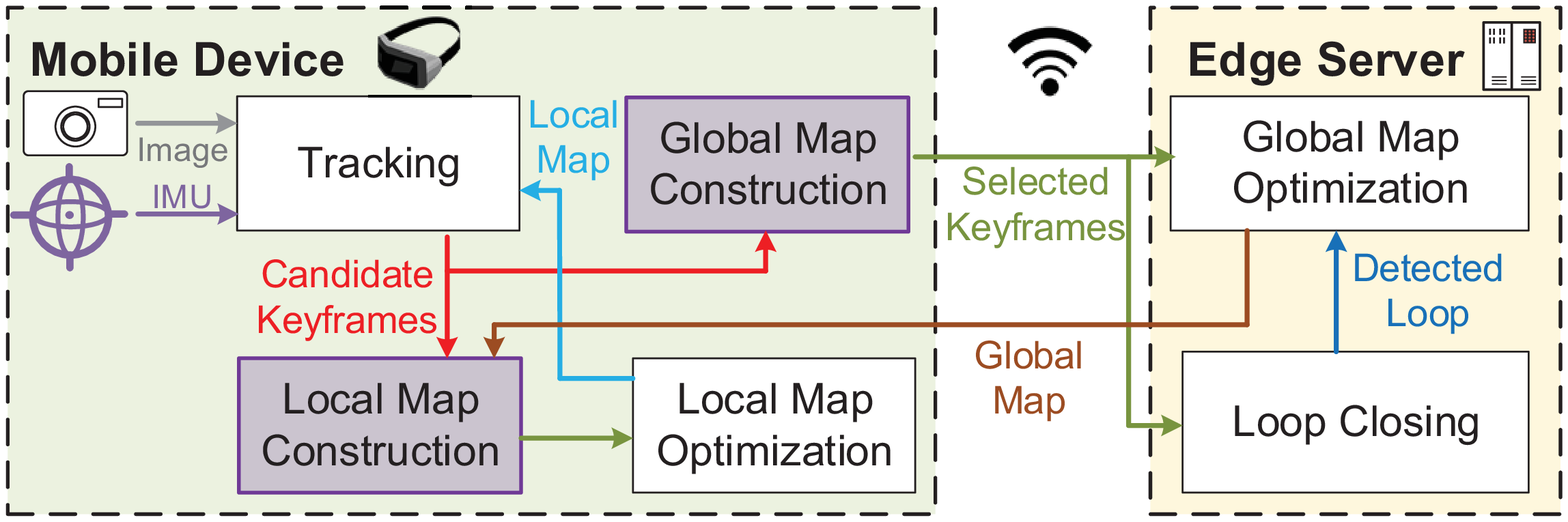}
       \vspace{-0.3cm}
  \caption{Overview of the AdaptSLAM system  architecture. 
  }
   \vspace{-0.08cm}
     \label{architecture}
     \end{figure}

\section{AdaptSLAM System Model}
\label{systemmodel2}

\subsection{The Pose Graph and the Map}
\label{sec:map}

We divide time into slots
of equal size of $\Delta t$. We  introduce the pose graph and the map at time slot $t$ that 
lasts for $\Delta t$ seconds. For clarity of notation, we will omit the time index below.

\begin{definition}[Pose graph]
\label{def:posegraph}
For a given index set $\mathcal{K} = \{1, \ldots, |\mathcal{K}|\}$ (indexing camera poses and representing keyframes), the pose graph is defined as the undirected multigraph $\mathcal{G} = \left(\mathcal{K}, \mathcal{E}, \mathcal{C} \right)$, where $\mathcal{K}$ is the node set, $\mathcal{E}$ is the edge set, and $\mathcal{C}=\left\{ \mathsf{IMU},\mathsf{vis}\right\}$ is the category set. Here, $\mathsf{IMU}$  stands for the IMU edges, and  $\mathsf{vis}$ stands for the covisibility edges.
\end{definition}

Given a pose graph $\mathcal{G} = \left(\mathcal{K}, \mathcal{E}, \mathcal{C} \right)$, there is a camera pose ${\bf P}_n = \left(x,y,z,w_x,w_y,w_z\right)$ for all $n \in \mathcal{K}$, where the first three entries are the 3-D positions and the last three ones are the Euler angles (yaw, pitch and roll) representing the camera orientation. Edges in $\mathcal{E}$ are represented as $e = ((n,m),c)$ for $n, m \in \mathcal{K}$ and $c\in  \mathcal{C}$. Two keyframes in $\mathcal{K}$ are connected by a covisibility edge if there are 3D map points observed in both keyframes. Two consecutive keyframes are connected by an IMU edge if there are accelerometer and gyroscope readings from one keyframe to another. 
There may exist both a covisibility edge and an IMU edge between two keyframes. 

For each $e = ((n,m),c) \in \mathcal{E}$, we 
observe relative noisy pose measurements between $n$ and $m$, which is written as 
$
\label{noisyMeasurement}
{\Delta _e} = {\bf P}_m - {{\bf P}_n} + {\bf x}_e
$, 
where ${\bf x}_e$ is the measurement noise on edge $e$. The map optimization  problem is to find the maximum likelihood estimates $\{{\bf{\tilde P}}_n\}_{n \in \mathcal{K}}$ for the actual camera poses $\{{\bf P}_n\}_{n \in \mathcal{K}}$. For Gaussian distributed edge noise, the map optimization problem is 
\begin{equation}
\label{eq:PGO}
    \mathop {\min }\limits_{\left\{{\bf \tilde{P}}_n\right\}_{n \in \mathcal{K}}} \sum\limits_{e \in {{\cal E}}} {{{\left( {{{\bf{\tilde{x}}}_e}} \right)}^ \top }{{\cal I}_e}{{\bf{\tilde{x}}}_e}}, 
\end{equation} 
where ${\bf \tilde{x}}_e = {\Delta _e} - {\bf{\tilde P}}_m + {\bf{\tilde P}}_n$ and $\mathcal{I}_e$ is the information matrix (i.e., inverse covariance matrix) of the measurement error on 
$e $~\cite{carlone2016planar}.  $\left({\bf \tilde{x}}_e\right)^\top\mathcal{I}_e {\bf \tilde{x}}_e$ is the Mahalanobis norm~\cite{orbslam3,vins} of the {\em estimated} measurement noise for $e$ with respect to  $\mathcal{I}_e$.

Below,  we 
assume that 
the measurement noise $\mathbf{x}_{e}$ is Gaussian distributed with isotropic covariance (as in~\cite{Placed21Fast,khosoussi2019reliable,7487264}). 
We assume that the information matrix ${\mathcal{I} _{e}}$ can be characterized by a weight assigned to $e$~\cite{Placed21Fast,cramer}. Specifically, ${ {\mathcal{I} _{e}} } = w_{e}{\mathcal{I} }$, where $w_{e} \geqslant 1$ is  the weight for $e$ and ${\mathcal{I}}$ is the matrix that is constant for all measurements. We note that the relative  measurements between keyframes $n$ and $m$ 
introduce the same information for them. 
We assume  all weights $w_e$ to be independent from each other for edges between different pairs of keyframes as in~\cite{orbslam3,cramer,boumal2014cramer,vins}. 

The map optimization problem in \eqref{eq:PGO} is solved  by standard methods such as Levenberg-Marquardt algorithm implemented in
g2o \cite{kummerle2011g} and Ceres solvers \cite{ceres-solver} as in \cite{orbslam3,vins}.

\begin{definition}[Anchor]
We say that a node is the anchor of the pose graph if the pose of the node is known.
\end{definition}

The map (local or global) consists of the pose graph (in Definition~\ref{def:posegraph}) and map points in the environment. In this paper, we will use the terms map and pose graph interchangeably. Without loss of generality, we will also assume that the global (or local) map is anchored on the first node, as in~\cite{cramer,Placed21Fast}. This assumption is made because SLAM can only estimate the relative pose change based on the covisibility and inertial measurements, while the absolute pose estimate in the global coordinate system cannot be provided.

\subsection{The Local Map} 
\label{sec:userFunc}

\textbf{Local map construction.}
The candidate keyframes are selected from camera frames according to the selection strategy in ORB-SLAM3 \cite{orbslam3}, and these candidate keyframes form the set $\mathcal{K}$. Due to the constrained computation resources, the mobile device  selects a fixed keyframe set $\mathcal{K}_{fixed}$ and a local keyframe set $\mathcal{K}_{loc}$ from the candidate keyframes, where $\left| {{{\cal K}_{fixed}}} \right|\leqslant l_{f}$ and $\left| {{{\cal K}_{loc}}} \right|\leqslant l_{loc}$. 
The fixed keyframe set $\mathcal{K}_{fixed} \subseteq \mathcal{K}_{g,user}$ is selected from the global map $\mathcal{K}_{g,user}$ that was last transmitted from the edge server. The poses of keyframes in  $\mathcal{K}_{fixed}$ act as fixed priors in the local map optimization. This is because poses of keyframes in  $ \mathcal{K}_{g,user}$  are already optimized in the global map optimization  and hence have low uncertainty. The poses of keyframes in the local keyframe set $\mathcal{K}_{loc}\subseteq \mathcal{K}\setminus\mathcal{K}_{g,user}$ will be optimized according to the map optimization problem introduced above.

The edges between keyframes in $\mathcal{K}_{loc}$ form the set $\mathcal{E}_{loc}$, and the edges whose one  node belongs to  $\mathcal{K}_{loc}$ and another node belongs to $\mathcal{K}_{fixed}$ form the set $\mathcal{E}_{l,f}$.

\textbf{Local map optimization.}
After selecting $\mathcal{K}_{loc}$ in the local map construction, the local map optimization is to optimize the estimated poses $\left\{\tilde{\bf{P}}_{n}\right\}_{n\in \mathcal{K}_{loc}}$ to minimize the sum of Mahalanobis norms $\sum\limits_{e \in {\mathcal{E}_{loc}} \cup {\mathcal{E}_{l,f}}} {{{\left( {{{\bf{\tilde{x}}}_e}} \right)}^ \top }{{\cal I}_e}{{\bf{\tilde{x}}}_e}} $. Note that 
in the local pose graph optimization, the keyframes in $\mathcal{K}_{fixed}$ are included in the optimization with their poses fixed. The local map optimization to solve \eqref{eq:PGO} is 

\begin{equation}
\label{localopt}\mathop {\min }\limits_{\{{\bf{\tilde{P}}}_n\}_{n \in {\mathcal{K}_{loc}}}} \sum\limits_{e \in {{\cal E}_{loc}} \cup {{\cal E}_{l,f}}} 
{{{\left( {{\bf{\tilde{x}}}_e} \right)}^\top}
\mathcal{I}_ e
{\bf{\tilde{x}}}_e}.
\end{equation}

\subsection{The Global Map}
\label{sec:serverFunc}

\textbf{Global map construction.}
Due to the limited bandwidth between the mobile device and the edge server, only a subset of candidate keyframes are offloaded to the edge server to build a global map. 
The selection of keyframes to be offloaded will be optimized to minimize the pose estimation uncertainty of the global map when considering the underlying wireless network constraints.

The edge server maintains the global map, denoted as $\mathcal{K}_{g,edge}$,
holding all keyframes uploaded by the mobile device. The edges between keyframes in the global map $\mathcal{K}_{g,edge}$ constitute the set $\mathcal{E}_{glob}$.  Note that $\mathcal{K}_{g,edge}$ may be different from $\mathcal{K}_{g,user}$, because 
the global map is large and it takes time to transmit the most up-to-date global map from the edge server to the mobile device.

\textbf{Global map optimization.}
After selecting $\mathcal{K}_{g,edge}$ in the global map construction, the edge server performs the global map optimization to estimate poses $\tilde{\mathbf{P}}_{n}$ in  $\mathcal{K}_{g,edge}$ and minimize the sum of Mahalanobis norms $\sum\limits_{e \in {{\cal E}_{glob}}}{{{\left( {{\bf{\tilde{x}}}_e} \right)}^\top}
\mathcal{I}  _e
{\bf{\tilde{x}}}_e}$. Specifically, the edge solves \eqref{eq:PGO} when $\mathcal{E}=\mathcal{E}_{glob}$ and $\mathcal{K}=\mathcal{K}_{g,edge}$, i.e., the global map optimization is to solve
\begin{equation}
\label{globalopt}
\mathop {\min }\limits_{{\{\bf{\tilde{P}}}_n\}_{n \in {\mathcal{K}_{g,edge}}}} \sum\limits_{e \in {{\cal E}_{glob}}}{{{\left( {{\bf{\tilde{x}}}_e} \right)}^\top}
\mathcal{I}_e
{\bf{\tilde{x}}}_e}.
\end{equation}

\section{Problem Formulation}
\label{sec:formulation}

AdaptSLAM aims to efficiently select keyframes to construct optimal local and global maps, i.e., we select keyframes in $\mathcal{K}_{loc}$ and $\mathcal{K}_{fixed}$ for the local map and $\mathcal{K}_{g, edge}$ for the global map.  
From \S\ref{systemmodel}, after constructing the optimal local and global maps, the map optimization can be performed using the standard algorithms \cite{kummerle2011g,ceres-solver}.
We construct optimal local and global maps by minimizing the uncertainty of  the keyframes' estimated poses. Hence, we represent and quantify the uncertainty in~\S\ref{sec:Uncertainty Quantification}, and  formulate the uncertainty minimization problems in \S\ref{sec:problemformulation}.

\subsection{Uncertainty Quantification}
\label{sec:Uncertainty Quantification}

Let $\mathbf{p}_n=\tilde{\mathbf{P}}_n-{\mathbf{P}}_n$ denote the pose estimate error of keyframe $n$. The estimated measurement noise can be rewritten as ${\widetilde {\mathbf{x}}_e} = {\mathbf{p}_n} - {\mathbf{p}_m} + {\mathbf{x} _e} =\mathbf{p}_{n,m}+ {\mathbf{x} _e}$, where $\mathbf{p}_{n,m}=\mathbf{p}_{n}-\mathbf{p}_{m}$. We stack all $\mathbf{p}_n, n\in\mathcal{K}$ and get a pose estimate error vector $\mathbf{w}=\left( {{{\bf{p}}^\top_1},{{\bf{p}}^\top_2}, \cdots ,{{\bf{p}}^\top_{\left| \mathcal{K} \right|}}} \right)$. We rewrite the objective function of map optimization in~\eqref{eq:PGO}  as $\sum\limits_{e \in \mathcal{E}} {{{\left( {{{{\bf{\tilde x}}}_e}} \right)}^ \top }{{\cal I}_e}{{{\bf{\tilde x}}}_e}}  = \sum\limits_{e = \left( {\left( {n,m} \right),c} \right) \in {\cal E}} {{\bf{p}}_{n,m}^ \top } {{\cal I}_e}{{\bf{p}}_{n,m}} + 2\sum\limits_{e = \left( {\left( {n,m} \right),c} \right) \in {\cal E}} {{\bf{p}}_{n,m}^ \top } {{\cal I}_e}{{\bf{x}}_e} + \sum\limits_{e \in {\cal E}} {{{\bf{x}}_e^\top}{{\cal I}_e}{{\bf{x}}_e}} $.
If we can rewrite the quadratic term $\sum\limits_{e = \left( {\left( {n,m} \right),c} \right) \in {\cal E}} {{\bf{p}}_{n,m}^ \top } {{\cal I}_e}{{\bf{p}}_{n,m}}$ in the format of  ${\mathbf{w}}{\mathcal{I}_w}{\mathbf{w}^\top}$, where ${\mathcal{I}_w}$ is called the information matrix of the pose graph, the uncertainty of the pose graph is quantified by $-\log \det \left( {{{\cal I}_w}} \right)$ according to the D-optimality~\cite{khosoussi2019reliable,carlone2018attention,chen2021anchor}.\footnote{Common approaches to quantifying uncertainty in SLAM are to use real scalar functions of the maximum likelihood estimator covariance matrix~\cite{rodriguez2018importance}. Among them, D-optimality (determinant of the covariance matrix)~\cite{Placed21Fast,cramer} 
captures the uncertainty due to all the elements of a covariance matrix and has well-known geometrical and information-theoretic interpretations~\cite{pukelsheim2006optimal}.}

We denote the pose estimate error vectors for the global and local maps as  ${{\bf{w}}_g} = \left( {{{\bf{p}}^\top_{{u_1}}},\cdots ,{{\bf{p}}^\top_{{u_{{|\mathcal{K}_{g,edge}|}}}}}} \right)$ and  
$\mathbf{w}_l=\left(  {\mathbf{p}^\top_{r_1}},\cdots,{\mathbf{p}^\top_{r_{\left|\mathcal{K}_{loc}\right|}}}\right)$, where $u_1,\cdots,{{u_{{|\mathcal{K}_{g,edge}|}}}}$ are the keyframes in $\mathcal{K}_{g,edge}$, 
and $r_1,\cdots,r_{\left|\mathcal{K}_{loc}\right|}$ are the keyframes in $\mathcal{K}_{loc}$.
The first pose in the global and local pose graph is known ($\mathbf{p}_{u_1}=0$, $\mathbf{p}_{r_1}=0$).
We rewrite the quadratic terms of the objective functions of global and local map optimizations in~\eqref{globalopt} and~\eqref{localopt}  as
$\sum\limits_{e = \left( {\left( {n,m} \right),c} \right) \in {{\cal E}_{glob}}} {{{ {{\bf{p}}_{n,m}^ \top }  }} {{\cal I}_e} {{{\bf{p}}_{n,m}} }}   =  {{\bf{w}}_g}{{\cal I}_{glob}}\left( {{{\cal K}_{g,edge}}} \right){\bf{w}}_g^ \top $ (or
$\sum\limits_{e = \left( {\left( {n,m} \right),c} \right) \in {{{\cal E}_{loc}} \cup {{\cal E}_{l,f}}}} {{\bf{p}}_{n,m}^ \top } {{\mathcal{I}}_e}{{\bf{p}}_{n,m}}=\mathbf{w}_l {{\cal I}_{loc}\left( {\mathcal{K}_{loc},\mathcal{K}_{fixed}} \right)} \mathbf{w}_l^\top$),
where  $\mathcal{I}_{glob}\left(\mathcal{K}_{g,edge}\right)$ and $ {{\cal I}_{loc}\left( {\mathcal{K}_{loc},\mathcal{K}_{fixed}} \right)}$ are called the information matrices of the global and local maps and will be derived later (in Definition~\ref{def:uncertainty} and Lemmas~\ref{lemma:global} and~\ref{lemma:local}).

\begin{definition}[Uncertainty]
\label{def:uncertainty} 

 The uncertainty of the global (or local) pose graph is defined as $-\log \det \left(
 \tilde{\mathcal{I}}_{glob}  \left(\mathcal{K}_{g,edge}\right)
\right)$ (or $-
\log \det \left( {\tilde{\cal I}_{loc}\left( {\mathcal{K}_{loc},\mathcal{K}_{fixed}} \right)}
\right)$, where $\tilde{\mathcal{I}}_{glob}\left({\mathcal{K}}_{g,edge}\right)$ and $ {\tilde{\cal I}_{loc}\left( {\mathcal{K}_{loc},\mathcal{K}_{fixed}} \right)}$ are obtained by removing the first row and first column in the information matrices $\mathcal{I}_{glob}\left(\mathcal{K}_{g,edge}\right)$ and $ {{\cal I}_{loc}\left( {\mathcal{K}_{loc},\mathcal{K}_{fixed}} \right)}$. 
\end{definition}

From Definition~\ref{def:uncertainty}, the uncertainty quantification is based on the global and local map optimizations introduced in \S\ref{sec:serverFunc} and~\S\ref{sec:userFunc}. After quantifying the uncertainty, we will later (in \S\ref{sec:problemformulation}) optimize the   local and global map construction which in turn minimizes the uncertainty of poses obtained from local and global map optimizations.

\begin{lemma}[Uncertainty of global pose graph]
\label{lemma:global}
For the global map optimization, the uncertainty is calculated as $-\log\det \left( \tilde{\mathcal{I}}_{glob}  \left(\mathcal{K}_{g,edge}\right)  \right)$, where 
$
\label{eq:globalUnc}
 \tilde{\mathcal{I}}_{glob}  \left(\mathcal{K}_{g,edge}\right)  = {\tilde{\mathbf L}_{glob}} \otimes \mathcal{I} 
$
with
$\tilde{\mathbf L}_{glob}$
being the matrix obtained by deleting the first row and column in the Laplacian matrix ${\mathbf L}_{glob}$, and $\otimes$ being the Kronecker product.
The $i,j$-th element of ${\mathbf L}_{glob}$ is given by
\begin{equation}\label{eq:Lglob}\begin{aligned}
\left[{\mathbf L}_{glob}\right]_{i,j}=\left\{ {\begin{array}{*{20}{c}}
 - \sum\limits_{e = (( {u_{{i}}},{u_{{j}}}) ,c)\in {\mathcal{E}_{g,edge}}} {{w_e}},\ \ \ {i\ne j} \\
\sum\limits_{e = (( {u_{{i}}},q),c)\in {\mathcal{E}_{g,edge}},u_i \ne {q} } {{w_e}},\ \ {i=j}
\end{array}} \right..\end{aligned}\end{equation}
\end{lemma}
\begin{proof}
See Appendix~\ref{proof:global}. Proof sketch:
The proof follows from  
the global map optimization formulation in \S\ref{sec:serverFunc}  and the definition of  $\tilde{\mathcal{I}}_{glob}\left(\mathcal{K}_{g,edge}\right)$.
\end{proof}

From Lemma~\ref{lemma:global}, the uncertainty of the global pose graph can be calculated based on the reduced Laplacian matrix (${\tilde{\mathbf{L}}}_{glob}$). According to the relationship between the reduced Laplacian matrix and the tree structure \cite{khosoussi2014novel}, the uncertainty is inversely proportional to  the logarithm of weighted number of spanning trees in the global pose graph. Similar conclusions are drawn for 2D pose graphs~\cite{khosoussi2019reliable} and 3D pose graphs with only covisibility edges~\cite{cramer,Placed21Fast}, where the device can move in 2D plane and 3D space respectively. We extend the results to VI-SLAM where the global pose graph is a multigraph with the possibility of having  both a covisibility edge and an IMU edge between two  keyframes.

\begin{lemma}[Uncertainty of local pose graph]
\label{lemma:local}
 The uncertainty is $-\log \det \left( \tilde{\cal I}_{loc}\left( {\mathcal{K}_{loc},\mathcal{K}_{fixed}} \right)\right)$ for the local map, where 
$
\label{eq:sigmaloc}
 \tilde{\cal I}_{loc}\left( {\mathcal{K}_{loc},\mathcal{K}_{fixed}} \right)= {\tilde{\mathbf L}_{loc}} \otimes \mathcal{I}
$
with $\tilde{\mathbf L}_{loc}$
being the matrix obtained by deleting the first row and the first 
column 
in ${\mathbf L}_{loc}$.
The $i,j$-th element of ${\mathbf L}_{loc}$ (of size $\left|\mathcal{K}_{loc}\right|\times \left|\mathcal{K}_{loc}\right|$) is given by
\begin{equation}
\label{eq:Lloc}\begin{aligned}
[{\mathbf L}_{loc}]_{i,j}=\left\{ {\begin{array}{*{20}{c}}
 - \sum\limits_{e = \left( {\left( {{r_{{i}}},{r_{j}}} \right),c} \right)\in {\mathcal{E}_{loc}}} {{w_e}}  ,\;\;\;\;\;\;\;\;\;\;{i} \ne {j}\\
{{\sum\limits_{e = \left( {\left( {{r_{{i}}},{q}} \right),c} \right)\in{\mathcal{E}_{l,f}} \cup {\mathcal{E}_{loc}},q \ne r_{i}} {{w_e}} } ,\;{i} = {j}}
\end{array}} \right..\end{aligned}\end{equation}
\end{lemma}
\begin{proof}
 See Appendix~\ref{proof:local}.
Proof sketch: Setting $\mathbf{p}_n=0$, $n\in\mathcal{K}_{fixed}$ (the fixed keyframes have poses with `zero uncertainty'), the proof follows from  the local pose graph optimization formulation in \S\ref{sec:userFunc} and the definition of ${\tilde{\cal I}_{loc}\left( {\mathcal{K}_{loc},\mathcal{K}_{fixed}} \right)}$.
\end{proof}

From Lemma~\ref{lemma:local}, the uncertainty of the local map is proportional to the uncertainty of the pose graph $\mathcal{G}$  anchoring on
the first node in $\mathcal{K}_{loc}$ and all nodes in $\mathcal{K}_{fixed}$, where $\mathcal{G}$'s node set is $\mathcal{K}_{fixed}\cup\mathcal{K}_{loc}$ and  edge set includes all measurements between any two nodes in $\mathcal{K}_{fixed}\cup\mathcal{K}_{loc}$. Note that keyframe poses in $\mathcal{K}_{fixed}$ are optimized on the edge server and transmitted to the mobile device, and they are considered as constants in the local pose graph optimization. From the uncertainty's perspective, adding fixed keyframes in $\mathcal{K}_{fixed}$ is equivalent to anchoring these keyframe poses (i.e., deleting rows and  columns corresponding to the anchored nodes in the Laplacian matrix of graph $\mathcal{G}$). In addition, from Lemma~\ref{lemma:local}, although poses are fixed, the anchored nodes still reduce the uncertainty of the pose graph. Hence, apart from $\mathcal{K}_{loc}$, we will select the anchored keyframe set $\mathcal{K}_{fixed}$ to minimize the uncertainty.

\subsection{Uncertainty Minimization Problems}
\label{sec:problemformulation}
We now formulate optimization problems whose objectives are to minimize the uncertainty of the  local and global maps. For the local map optimization, under the computation resource constraints, we solve Problem~\ref{prb1} for each keyframe~$k$. 
For the global map optimization, under the communication resource constraints, we solve Problem~\ref{prb2} to adaptively offload keyframes to the edge server.

\begin{problem}[Local map construction]
\label{prb1}
\begin{align}
&\mathop {\max }\limits_{{{\cal K}_{loc}},{{\cal K}_{fixed}}} \log \det \left( {{{\widetilde I}_{loc}}\left( {{{\cal K}_{loc}} \cup \left\{ k \right\},{{\cal K}_{fixed}}} \right)} \right)\\
&\textrm{s.t.}\ \ 
 \label{c2local}\left| {\mathcal{K}_{loc}} \right|\leqslant {l_{loc}}, \mathcal{K}_{loc} \subseteq \mathcal{K}\setminus\mathcal{K}_{g,user}\\&
\label{c3local}\left| {\mathcal{K}_{fixed}} \right|\leqslant {l_f},\mathcal{K}_{fixed}\subseteq \mathcal{K}_{g,user}.
\end{align}
\end{problem}
\noindent The objective of Problem~\ref{prb1} is equivalent to minimizing the uncertainty of the local map. Constraint~\eqref{c2local}  means that the size of ${\mathcal{K}_{loc}} $ is constrained to reduce the computational complexity in the local map optimization, and that the keyframes to be optimized in the local map are selected from keyframes that are not in $\mathcal{K}_{g,user}$. Constraint~\eqref{c3local} means that the size of  ${\mathcal{K}_{fixed}} $ is constrained, and that the fixed keyframes are selected from $\mathcal{K}_{g,user}$ that were previously optimized on and transmitted from the edge server.

\begin{problem}[Global map construction]
\label{prb2}
\begin{align}
&\mathop {\max }\limits_{{\cal K}^\prime \subseteq {\cal K} \setminus {{\cal K}_{g,edge}}} \log \det \left( {{{\widetilde {\cal I}}_{glob}}\left( {{{\cal K}_{g,edge}}\cup {\mathcal{K}^\prime}} \right)} \right)\\
&\textrm{s.t.}\ \ 
\label{c1global} d\left| {{\mathcal{K}^\prime }} \right|\leqslant{D}.
\end{align}
\end{problem}
\noindent The objective of Problem~\ref{prb2} is equivalent to  minimizing the uncertainty of the global map. ${\mathcal{K}}\setminus\mathcal{K}_{g,edge}$ is set of the keyframes that have not been offloaded to the server, and we select a subset of keyframes, $\mathcal{K}^\prime$, from ${\mathcal{K}}\setminus\mathcal{K}_{g,edge}$. The constraint~\eqref{c1global} guarantees that the keyframes cannot be offloaded from the device to the server at a higher bitrate than the available channel capacity, where $D$ is the channel capacity constraint representing the maximum number of bits that can be transmitted in a given transmission window.   
We assume that the data size $d$ of each keyframe is the same, which is based on the observation that the data size is relatively consistent across keyframes in popular public SLAM
datasets~\cite{burri2016euroc,TUMdataset}.

\section{Local and Global Map Construction}
\label{sec:ARproblem1}

\label{sec:ARproblem}

We analyze the properties of approximate submodularity in map construction problems, and
propose low-complexity algorithms to efficiently construct local and global maps.  

\subsection{Local Map Construction}
\label{sec:LocalProblem}
 The keyframes in the local map include those in two disjoint sets ${\mathcal{K}_{loc}} $ and ${\mathcal{K}_{fixed}}$. To efficiently solve Problem~\ref{prb1}, we decompose it into two problems aiming at minimizing the uncertainty: Problem~\ref{prb3} that selects keyframes in ${\mathcal{K}_{loc}} $ and Problem~\ref{prb4}
that selects keyframes in ${\mathcal{K}_{fixed}}$. We obtain the optimal local keyframe set ${\mathcal{K}_{loc}^{\star}}$ in Problem~\ref{prb3}. Based on ${\mathcal{K}_{loc}^{\star}}$, we then obtain the optimal fixed keyframe set ${\mathcal{K}_{fixed}^{\star}}$ in Problem~\ref{prb4}.
We will compare the solutions to Problems~\ref{prb3} and~\ref{prb4} with the optimal solution to Problem~\ref{prb1} in \S\ref{sec:evaluation} to show that the performance loss induced by the decomposition is small.
\begin{problem}
\label{prb3}
\begin{align}
&\nonumber {\mathcal{K}_{loc}^{\star}}=\arg  \mathop {\max }\limits_{\mathcal{K}_{loc}} \log \det \left( {\tilde{\cal I}_{loc}{{\left( {\mathcal{K}_{loc}{ \cup \left\{ k \right\}},\emptyset } \right)}}} \right)\\
&\textrm{s.t.}\ \ \nonumber
\eqref{c2local}.
\end{align}
\end{problem}

\begin{problem}
\label{prb4}
\begin{align}
&\nonumber {\mathcal{K}_{fixed}^{\star}}=\arg\mathop {\max }\limits_{\mathcal{K}_{fixed}} \log \det \left( {\tilde{\mathcal I}_{loc}{{\left( {\mathcal{K}_{loc}^{\star}{ \cup \left\{ k \right\}},\mathcal{K}_{fixed}} \right)}}} \right)\\
&\textrm{s.t.}\ \ 
\nonumber
\eqref{c3local}.
\end{align}
\end{problem}

\subsubsection{The Selection of Local Keyframe Set $\mathcal{K}_{loc}$}
\label{sec:optimizedKeyframe}
We first solve Problem~\ref{prb3}. It is a  nonsubmodular optimization problem with constraints,  which are NP-hard and generally difficult to be solved with an approximation ratio \cite{bian2017guarantees}. Hence, we  decompose Problem~\ref{prb3} into  subproblems (Problems~\ref{subproblem1} and~\ref{subproblem2}) that are equivalent to the original Problem~\ref{prb3} and can be approximately solved with a low-complexity algorithm.

In problem~\ref{subproblem1}, assume that we already select a keyframe subset  $\mathcal{K}_{base}$ from $\mathcal{K}\setminus\mathcal{K}_{g,user}$ (with the size $l_{b}\triangleq \left|\mathcal{K}_{base}\right|\leqslant l_{loc}$), and we aim to further select a keyframe set $\mathcal{K}_{add}$ to be added to $\mathcal{K}_{base}$ to minimize the local map uncertainty. Rewriting the objective as $\mathsf{Unc}\left(\mathcal{K}_{add}\cup {\mathcal{K}_{base}}\cup\{k\}\right )\triangleq-\log \det \left( {{{\widetilde {\cal I}}_{loc}}\left(\mathcal{K}_{add}\cup {\mathcal{K}_{base}}\cup\{k\},\emptyset \right)}\right)$, the problem is to obtain the optimal $\mathcal{K}_{add}$ (denoted as $\mathsf{OPT}_{add}(\mathcal{K}_{base})$) given $\mathcal{K}_{base}$:

\begin{problem}
\label{subproblem1}
\begin{align}
\nonumber\label{sub_obj1}&\mathsf{OPT}_{add}\left( {{{\cal K}_{base}}} \right) = \arg \mathop {\max }\limits_{{{\cal K}_{add}}} -\mathsf{Unc}\left(\mathcal{K}_{add}\cup {\mathcal{K}_{base}}\cup\{k\}\right )\\
&\nonumber\textrm{s.t.} \ \ \left|\mathcal{K}_{add}\right|\leqslant l_{loc}-l_b.
\end{align}
\end{problem}

After getting the solutions (i.e., $\mathsf{OPT}_{add}\left( {{{\cal K}_{base}}} \right)$) to Problem~\ref{subproblem1} for all possible $\mathcal{K}_{base}$ of size $l_b$, we obtain the optimal $\mathcal{K}_{base}$ (denoted as $\mathcal{K}_{base}^\star$) in Problem~\ref{subproblem2}.
\begin{problem}
\label{subproblem2}
\begin{align}
&\nonumber{\cal K}_{base}^\star= \arg\mathop {\max }\limits_{{{\mathcal K}_{base}}} -\mathsf{Unc}\left(\mathsf{OPT}_{add}(\mathcal{K}_{base})\cup\mathcal{K}_{base} \cup \{k\}\right)\\
&\nonumber\textrm{s.t.} \ \ \left|\mathcal{K}_{base}\right|= l_b.
\end{align}
\end{problem}

\begin{lemma}
\label{lemma:equiproblem}
Given $l_b$, the solution to Problems~\ref{subproblem1} and~\ref{subproblem2}, i.e., ${\cal K}_{base}^\star$ and $\mathsf{OPT}_{add}\left( {{{\cal K}_{base}^\star}} \right)$, will give us the solution ${\mathcal{K}_{loc}^{\star}}$ to Problem~\ref{prb3}. Specifically, ${\mathcal{K}_{loc}^{\star}}={\cal K}_{base}^\star \cup \mathsf{OPT}_{add}\left( {{{\cal K}_{base}^\star}} \right)$.
\end{lemma}

\begin{proof}
The proof is straightforward and hence omitted.
\end{proof}

We can obtain ${\mathcal{K}_{loc}^{\star}}$ in Problem~\ref{prb3} by solving Problems~\ref{subproblem1} and~\ref{subproblem2}. We will show that the objective function of Problem~\ref{subproblem1} 
is `close to' a submodular function when the size of the keyframe set $\mathcal{K}_{base}$ is large. In this case, Problem~\ref{subproblem1} can be efficiently solved using a greedy algorithm with an approximation  ratio. When  $\left|\mathcal{K}_{base}\right|$ is small, 
 we need to compare the objective function for different combinations of $\mathcal{K}_{base}$ and $\mathcal{K}_{add}$.

\begin{lemma}
\label{bound}
When  $\frac{{{w_{\max }}}}{{\left| {{{\cal K}_{base}}} \right|{w_{\min }}}}<1$, the submodularity ratio $\gamma$ of the objective function in Problem~\ref{subproblem1} is lower bounded by 
\begin{equation}\gamma \geqslant 1 + \frac{1}{\vartheta }\log \left( {1 - \frac{{4{{\left| {{{\mathcal K}_{add}}} \right|}^2}{w^2_{{{\max }}}}}}{{\left| {{{\cal K}_{base}}} \right|{w_{\min }} - {w_{\max }}}}} \right),
\end{equation}
where 
$\vartheta  = \mathop {\min }\limits_{{m} \in  {{\mathcal{K}_{add}}} } \sum\limits_{{n} \in {{\mathcal{K}_{base}}} } {\log {w_{{n},{m}}}}$, 
${w_{\max }} = \mathop {\max }\limits_{{n}  ,{m} \in  {{{\cal K}_{base}} \cup {{\cal K}_{add}}} } {w_{{n},{m}}}$, and
${w_{\min }} = \mathop {\min }\limits_{{n},{m} \in  {{{\cal K}_{base}} \cup {{\cal K}_{add}}} } {w_{{n},{m}}}.$
$\gamma$ is close to 1 when $ \left| {{{\cal K}_{base}}} \right|$ is significantly larger than $\left| {{{\cal K}_{add}}} \right|$.
\end{lemma}
\begin{proof} See Appendix~\ref{proof:bound}. Proof sketch:
Following from the definition of $\gamma$ in~\eqref{eq:ratio},
we first prove that the denominator  in~\eqref{eq:ratio}, denoted as $\log\det(\mathbf{M}_{den})$, is lower bounded by $\vartheta$. Denoting the numerator in~\eqref{eq:ratio} as $\log\det(\mathbf{M}_{num})$,
 we show that $\log \det ({{\bf{M}}_{num}}) \geqslant \log\det(\mathbf{M}_{den})  + \log \left( {1 - \frac{{4{{\left| {{{\cal K}_{add}}} \right|}^2}w_{\max }^2}}{{\left| {{{\cal K}_{base}}} \right|{w_{\min }} - {w_{\max }}}}} \right)$, by proving that the absolute values of all elements in $\mathbf{M}_{num}$ are bounded. 
\end{proof}

 From Lemma~\ref{bound}, the objective function in Problem~\ref{subproblem1} is `close to' a submodular function when  the size of the existing keyframe set (i.e., $\left|\mathcal{K}_{base}\right|$) is much larger than $\left|\mathcal{K}_{add}\right|$. 
 Hence, 
 we can use the greedy algorithm 
 to approximately solve Problem~\ref{subproblem1}. According to Theorem~\ref{theorem:submax},  the 
solution obtained by the 
greedy algorithm for Problem~\ref{subproblem1}, denoted by $\mathsf{OPT}_{add}^\#\left( {{{\cal K}_{base}}} \right)$, has an approximation guarantee that $\mathsf{OPT}_{add}^\#\left( {{{\cal K}_{base}}} \right)\geqslant \left( {1 - \exp ( - \gamma )} \right)\mathsf{OPT}_{add}\left( {{{\cal K}_{base}}} \right)$.
 
According to the analysis of the properties of Problems~\ref{subproblem1} and~\ref{subproblem2}, we now solve Problem~\ref{prb3} to select the local keyframe set $\mathcal{K}_{loc}$ using Algorithm~\ref{alg_opt} (top-$h$ greedy-based algorithm). $\Theta$ is the set of possible keyframe sets that minimize the local map uncertainty, and we only maintain $h$ keyframe sets to save the computation resources. $\Lambda, \Lambda\in\Theta$, denotes the element in $\Theta$ and represents one possible keyframe set. When the size of $\Lambda$ is smaller than a threshold $l_{thr}$ ($|\Lambda|\leqslant l_{thr}$), we select the top-$H$ ($H>1$) highest-scoring combinations of  $\Lambda$ and $n,n\in{{\mathcal{K}} \setminus  {\mathcal{K}_{g,user}} }$, that minimize $\mathsf{Unc}\left( {\Lambda}\cup{\left\{ k,n\right\}} \right)$. 
 When $\left|\Lambda\right|$ gets larger, we only select the highest-scoring combination. The reasons are as follows.  $\Lambda$ can be seen as the existing keyframe set $\mathcal{K}_{base}$. According to Lemma~\ref{bound}, when the size of the existing keyframe set (which is $\left|\Lambda\right|$ here) is small, there is no guarantee that $\mathsf{Unc}\left(\mathcal{K}_{add}\cup{\mathcal{K}_{base}}\cup\{k\}\right)$ is close to a submodular function (i.e., the submodularity ratio is much smaller than~1). Hence, we need to try different combinations of  $\Lambda$ and $n$ to search for the combination that minimizes the uncertainty after each iteration. As $\left|\Lambda\right|$ grows, the submodularity ratio  is close to~1, and 
  a greedy algorithm can achieve $\eta$ approximation ($\eta=1-\exp(-\gamma)$, $\gamma\to1$). In this case, we apply the greedy algorithm and only keep the combination that achieves the minimal uncertainty at each step.

 \begin{algorithm}[t]
 \caption{Selecting local keyframe set $\mathcal{K}_{loc}$ in the local map (top-$h$ greedy-based algorithm)}
 \begin{algorithmic}[1]
  \label{alg_opt}
\STATE $\Theta \leftarrow\emptyset$; 
 \WHILE{( $|\Lambda|\leqslant l_{loc}$)}
 \STATE \algorithmicif { $\text{ }|\Lambda| \leqslant l_{thr}$} \algorithmicthen   $\text{ }h\leftarrow H$ \algorithmicelse   $\text{ }h\leftarrow 1$;
\STATE \label{line8} Select the top-$h$ highest-scoring combinations of $\Lambda, \Lambda\in \Theta$ and $n,n\in{{\mathcal{K}} \setminus  {\mathcal{K}_{g,user}} }$ that minimize $\mathsf{Unc}\left( {\Lambda}\cup{\left\{ n,k\right\}} \right)$. $\mathsf{Unc}\left( {\Lambda}\cup{\left\{ n,k\right\}} \right)$ is calculated using the computation reuse algorithm  in Algorithm~\ref{alg_opt};
\STATE Update $\Theta$ as the set of $h$ highest-scoring combinations of $\Lambda$ and $n$. Each element of $\Theta$ is a set (i.e., $\Lambda\cup\left\{n\right\}$) corresponding to one combination;
  \ENDWHILE
\STATE  ${\cal K}_{loc}^ \star  \leftarrow \arg \mathop {\min }\limits_{\Lambda  \in \Theta } \mathsf{Unc}(\Lambda\cup\{k\}) $.
 \end{algorithmic} 
 \end{algorithm}

 \subsubsection{Computation Reuse Algorithm}
\label{sec:computereuse}

\begin{algorithm}[t]
 \caption{Computation reuse algorithm}
 \begin{algorithmic}[1]
  \label{alg_computation_reuse}
 \STATE  \textbf{Input:} $\det(\mathbf{A})$, $\mathbf{A}^{-1}$; 
 \STATE $\mathbf{B}\leftarrow{{{\bf{A}}^{ - 1}}}$. Calculate ${{\mathbf{B}_i}\mathbf{B}_i^\top}$, $i=1,\cdots,|\Lambda|$;
 \STATE Calculate ${\left( {{{\bf{A}}^\prime }} \right)^{ - 1}}$ using \eqref{eq:Ainverse}. Calculate $\det \left( {{{\bf{A}}^\prime }}\right)$ using \eqref{eq:detA}.
  Calculate $\det(\tilde{\mathcal{I}}\left( {\Lambda  \cup \{ n,k\} } \right))$ using \eqref{eq:detFinal}.
 \end{algorithmic} 
 \end{algorithm}
\vspace{-0em}

We use the computation reuse algorithm (Algorithm~\ref{alg_computation_reuse}) to speed up Algorithm~\ref{alg_opt}. We observe that for different $n, n\in \mathcal{K}\setminus\mathcal{K}_{g,user}$, only a limited number ($3|\Lambda|+1$) of elements in the matrix $\tilde{\mathcal{I}}\left( {\Lambda  \cup \{ n,k\} } \right)$ are different. Calculating the log-determinant function of a $(|\Lambda|+1)\times(|\Lambda|+1)$ matrix $\tilde{\mathcal{I}}\left( {\Lambda  \cup \{ n,k\} } \right)$ has a high computational complexity (of $\mathcal{O}(|\Lambda|+1)^3$)~\cite{strang2006linear}.
Hence, instead of computing the objective function
for each $n$ 
from scratch, we reuse parts of computation results for different $n$.

Letting $\mathbf{A}\triangleq\tilde{\mathcal{I}}\left( {\Lambda  \cup \{ k\} } \right)$ denote the information matrix of the local map in the $|\Lambda|$-th iteration (of Algorithm~\ref{alg_opt}), the information matrix in the $(|\Lambda|+1)$-th iteration is 
$\tilde{\mathcal{I}}\left( {\Lambda  \cup \{ n,k\} } \right) =\left[ {\begin{array}{*{20}{c}}
{\mathbf{A} + \rm{diag}\left( \mathbf{a} \right)}&\mathbf{a^\top}\\
{{\mathbf{a}}}&d
\end{array}} \right]$, where $\mathbf{a} = ({a_1},{a_2}, \cdots ,{a_{|\Lambda|}})$ with $a_i=w_{\lambda_i,n}$, $\lambda_i$ is the $i$-th element of $\Lambda$, and $d = {w_{k,n}} + \sum\limits_{i = 1}^{\left| \Lambda  \right|} {{a_i}}  $.

We aim to calculate $\det(\tilde{\mathcal{I}}\left( {\Lambda  \cup \{ n,k\} } \right))$ using the calculation of $\det(\mathbf{A})$ and ${\mathbf{A}}^{-1}$ from the previous iteration. Letting $\mathbf{A}^{\prime}\triangleq \mathbf{A}+\rm{diag}(\mathbf{a})$, $\det(\tilde{\mathcal{I}}\left( {\Lambda  \cup \{ n,k\} } \right))$ is calculated by \begin{equation}
\label{eq:detFinal}
\det(\tilde{\mathcal{I}}\left( {\Lambda  \cup \{ n,k\} } \right))=(d-\mathbf{a}(\mathbf{A^{\prime}})^{-1}\mathbf{a}^\top)\det(\mathbf{A}^{\prime}).
\end{equation}
Next we efficiently calculate  $(\mathbf{A^{\prime}})^{-1}$ and $\det(\mathbf{A}^{\prime})$ to get $\det(\tilde{\mathcal{I}}\left( {\Lambda  \cup \{ n,k\} } \right))$. We can rewrite  $\mathbf{A}^\prime$ as $\mathbf{A}^\prime=\mathbf{A}+\sum\limits_{i = 1}^{\left| \Lambda  \right|} {\beta _i^ \top } \beta _i$ where $\beta _i= \left( {0, \cdots ,\underbrace {{\sqrt{a_i}}}_{i{\rm{-th}}}, \cdots ,0} \right)$. According to Sherman–Morrison formula \cite{golub2013matrix}, $\left( {{{\bf{A}}^\prime }} \right)^{-1}$ is given by
\begin{equation}
\label{eq:Ainverse}
{\left( {{{\bf{A}}^\prime }} \right)^{ - 1}} \approx \underbrace {\mathbf{B}}_{{\rm{Reuse}}} - \sum\limits_{i = 1}^{\left| \Lambda  \right|} {\frac{{a_i}}{{1 + a_i{\mathbf{B}_{i,i}}}}\underbrace {{\mathbf{B}_i}\mathbf{B}_i^\top}_{{\rm{Reuse}}}},  
\end{equation}
where  $\mathbf{B}={{{\bf{A}}^{ - 1}}}$, $\mathbf{B}_{i,i}$ is the $i,i$-th element of $\mathbf{B}$, and $\mathbf{B}_{i}$ is the $i$-th column vector of $\mathbf{B}$. Using \eqref{eq:Ainverse}, $\mathbf{B}$ and ${{\mathbf{B}_i}\mathbf{B}_i^\top}$ can be computed
only once to be used for different $n, n\in \mathcal{K}\setminus \mathcal{K}_{g,user}$, which greatly reduces
the computational cost. According to the rank-1 update of determinant \cite{golub2013matrix}, $\det(\mathbf{A}^\prime)$ can be written as
\begin{equation}
\label{eq:detA}
\begin{aligned}
\det \left( {{{\bf{A}}^\prime }} \right) =&
\det \left( {\bf{A}} \right) \left( {1 + a_1{\bf{B}_{1,1}}} \right) \{\mathds{1}(\left| \Lambda  \right|=1)+\mathds{1}(\left| \Lambda  \right|>1)\\ \times &
\left.\prod\limits_{i = 2}^{\left| \Lambda  \right|} {\left( {1 + a_i\left[ {{\bf{B}} - \sum\limits_{j = 1}^{i - 1} {\frac{{a_j{{\bf{B}}_j}{\bf{B}}_j^T}}{{1 + a_j{{\bf{B}}_{j,j}}}}}} \right]_{i,i}} \right)}\right\}.
\end{aligned}
\end{equation}
$\left( {{\bf{B}} - \sum\limits_{j = 1}^{i - 1} {\frac{{a_j{{\bf{B}}_j}{\bf{B}}_j^T}}{{1 + a_j{{\bf{B}}_{j,j}}}}}} \right)$ is already calculated in \eqref{eq:Ainverse}, which reduces the computational complexity. Substituting~\eqref{eq:Ainverse} and~\eqref{eq:detA} into~\eqref{eq:detFinal}, we get the final results of $\det(\tilde{\mathcal{I}}\left( {\Lambda  \cup \{ n,k\} } \right))$.

\textbf{The computation complexity of different algorithms.} If we select keyframes in $\mathcal{K}_{loc}$ using a brute-force algorithm based on exhaustive enumeration of combinations of keyframes in $\mathcal{K}_{loc}$, the complexity is $\mathcal{O}\left({\rho \choose l_{loc}} l_{loc}^3\right)$, where $\rho =\left|\mathcal{K}\setminus\mathcal{K}_{g,user}\right|$ is the number of keyframes that have not been offloaded to the edge server. 
Without computation reuse, the computation complexity of the proposed top-$h$ greedy-based algorithm is $\mathcal{O}(H\rho l_{loc}^4)$. 
With computation reuse, it is reduced  to $\mathcal{O}(Hl_{loc}^4)+\mathcal{O}(H\rho l_{loc}^3)$. Since we only keep $l_{loc}$ keyframes in $\mathcal{K}_{loc}$ of the local map and a small $H$ in Algorithm~\ref{alg_opt} to save computation resources, i.e., $\rho \gg l_{loc}>H$,  the proposed greedy-based algorithm with computation reuse significantly reduces the computational complexity.

\subsubsection{The Selection of Fixed Keyframe Set $\mathcal{K}_{fixed}$}
\label{sec:fixedKeyframe}

 After selecting the local keyframe set $\mathcal{K}_{loc}$ by solving Problem~\ref{prb3}, we solve Problem~\ref{prb4} to select the fixed keyframe set. 
\begin{lemma}
\label{fixed_lemma}
Problem~\ref{prb4} is non-negative, monotone and submodular with a cardinality-fixed constraint. 
\end{lemma}
\begin{proof}[Proof sketch]
 It is straightforward to prove the non-negativity and monotonicity. For the submodularity, we can prove that $\frac{{\det \left( {{{\widetilde {\cal I}}_{loc}}\left( {{\cal K}_{loc}^ \star  \cup \left\{ k \right\},L} \right)} \right)\det \left( {{{\widetilde {\cal I}}_{loc}}\left( {{\cal K}_{loc}^ \star  \cup \left\{ k \right\},S} \right)} \right)}}{{\det \left( {{{\widetilde {\cal I}}_{loc}}\left( {{\cal K}_{loc}^ \star  \cup \left\{ k \right\},L \cup S} \right)} \right)\det \left( {{{\widetilde {\cal I}}_{loc}}\left( {{\cal K}_{loc}^ \star  \cup \left\{ k \right\},\emptyset } \right)} \right)}} \geqslant 1$, using the property that   $\det (\mathbf{M})\geqslant\det(\mathbf{N})$ holds for positive semidefinite matrices $\mathbf{M}$, $\mathbf{N}$ when $\mathbf{M}-\mathbf{N}$ is positive semidefinite.
\end{proof}

Lemma~\ref{fixed_lemma} indicates that the problem can be approximately solved with greedy methods in Algorithm~\ref{alg:greedy}~\cite{nemhauser1978analysis}.  For each iteration, the algorithm selects one keyframe  from $\mathcal{K}_{g,user}$ to be added to the fixed keyframe set $\mathcal{K}_{fixed}$.  
The
approximation ratio $\eta  = 1 - \exp(-1)$ guarantees that worst-case performance of a greedy algorithm cannot be far from optimal.

\subsection{Global Map Construction}
We use a low-complexity algorithm to solve Problem~\ref{prb2} to construct the global map. 
The objective function of Problem~\ref{prb2} can be rewritten as $-\mathsf{Unc}\left(\mathcal{K}_{g,edge}\cup\mathcal{K}^\prime\right)$, which has the  same structure as that of Problem~\ref{prb3}.  
Problems~\ref{prb2} and \ref{prb3} both add keyframes to the existing keyframe sets 
to construct a pose graph and optimize the keyframe poses in the pose graph. Hence, Algorithms~\ref{alg_opt} and~\ref{alg_computation_reuse} can be used to solve Problem~\ref{prb2}. In Algorithm~\ref{alg_opt}, $l_{loc}$ is replaced by $\frac{{{D}}}{d}$, and $\mathcal{K}\setminus\mathcal{K}_{g,user}$ is replaced by $\mathcal{K}\setminus\mathcal{K}_{g,edge}$. Calculating the uncertainty of a large global map is computationally intensive, and hence the proposed low-complexity algorithm is essential to reducing the computational load on the mobile device.

\label{sec:ARproblem2}

\section{Evaluation}
\label{sec:evaluation}

\begin{figure*}[t]
\begin{minipage}[t]{0.7\linewidth}
\vspace{-3.6cm}
\begin{subfigure}{.495\textwidth}
  \centering
   \includegraphics[width=\linewidth]{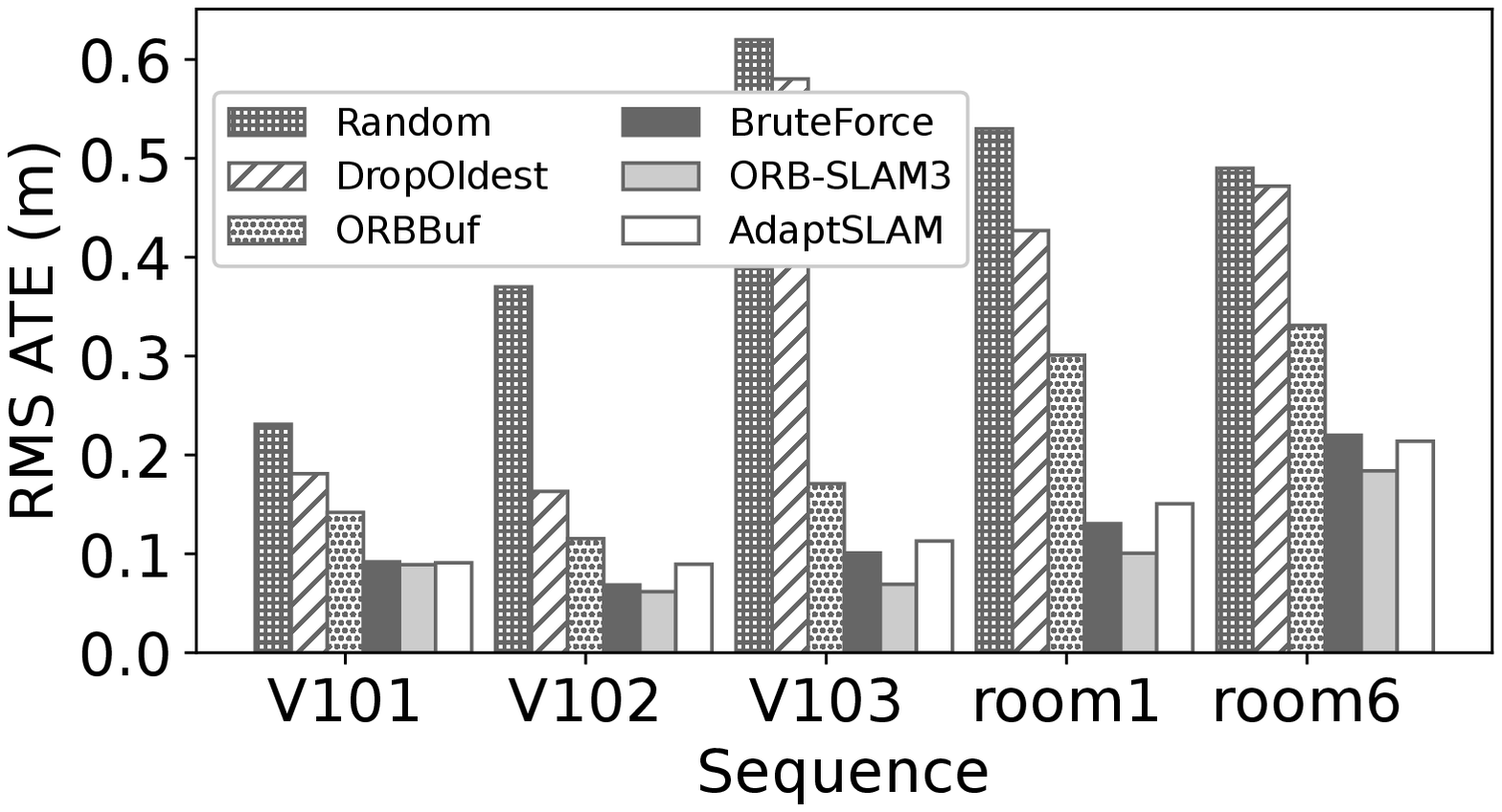}   
   \vspace{-0.6cm}
  \caption{V-SLAM}
  \label{VLocal}
\end{subfigure}
\begin{subfigure}{.495\textwidth}
  \centering
   \includegraphics[width=\linewidth]{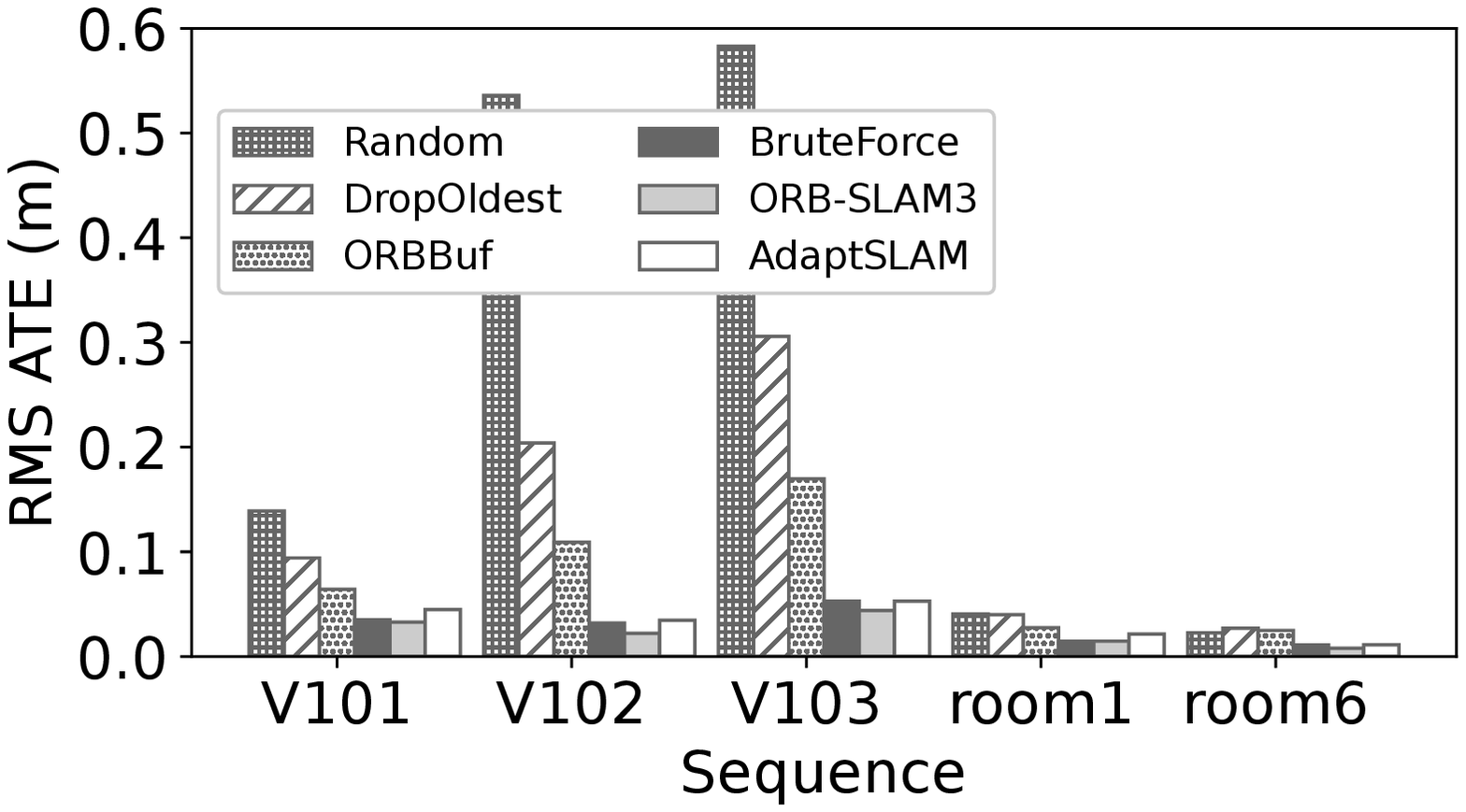}   
   \vspace{-0.6cm}
  \caption{VI-SLAM}
  \label{VILocal}
\end{subfigure}
\vspace{-0.15cm}
  \caption{RMS ATE for 6 keyframe selection methods in the local map construction for 5 sequences in EuRoC and TUM.
  }
     \label{Fig:Local}
\end{minipage}
\hspace{0.1cm}
\begin{minipage}[t]{0.29\linewidth}
\centering
   \includegraphics[width=0.9\textwidth]{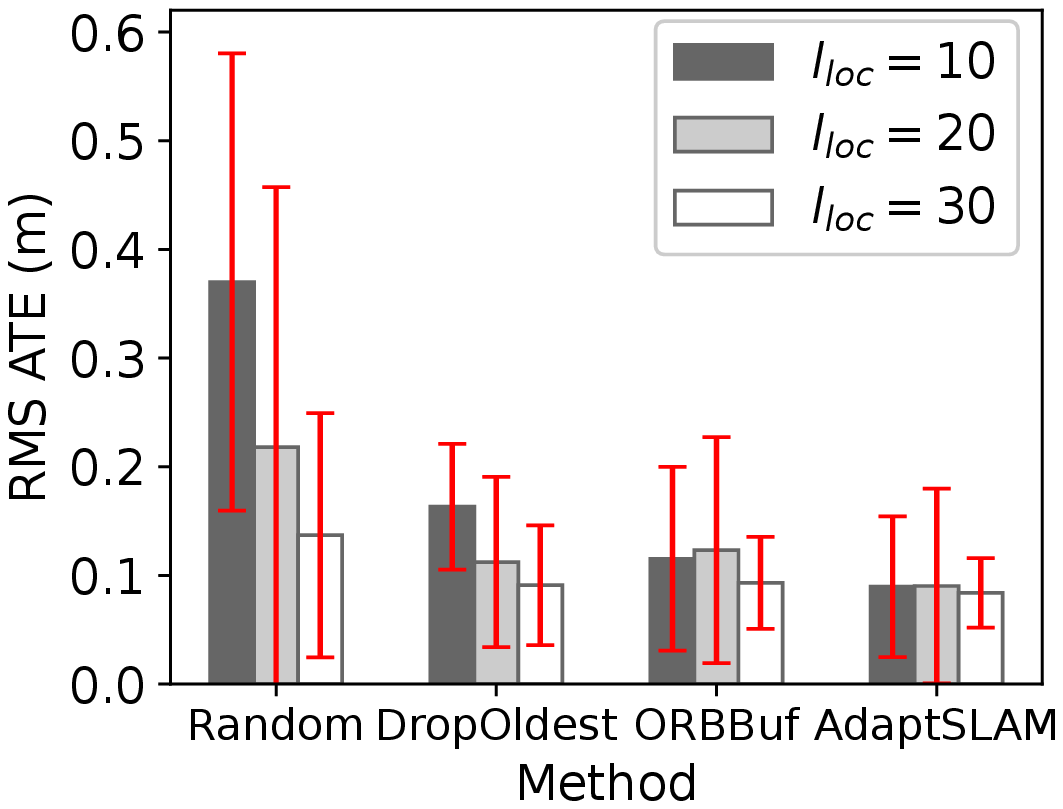}
      \vspace{-0.2cm}
  \caption{RMS ATE for different sizes of local keyframe set (for EuRoC V102). }
  \label{localsize}
\end{minipage}
\vspace{-0.5cm}
\end{figure*}

\begin{figure*}[hbt!]
\begin{minipage}[t]{0.7\linewidth}
\vspace{-3.6cm}
\begin{subfigure}{.495\textwidth}
  \centering
   \includegraphics[width=\linewidth]{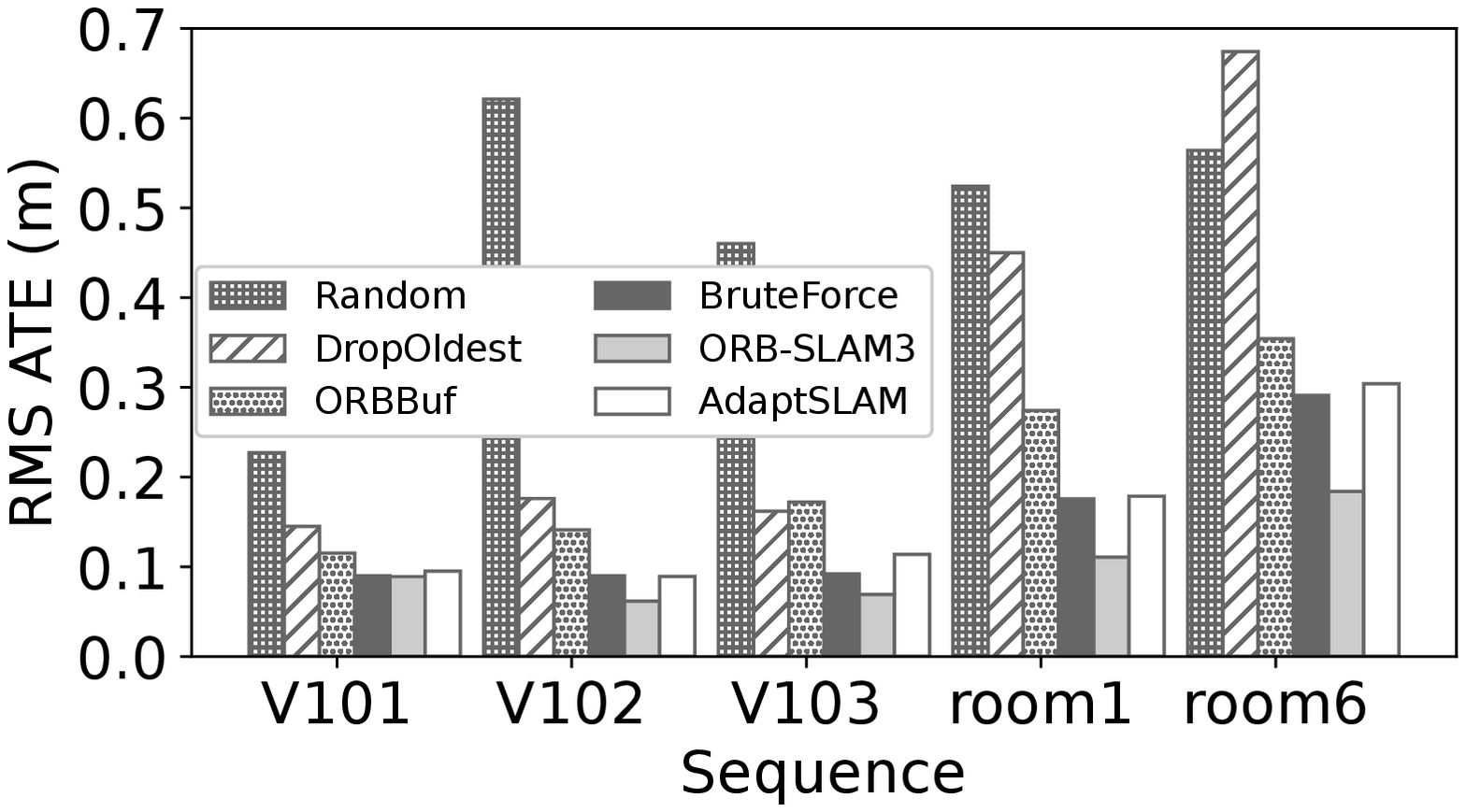}   
   \vspace{-0.6cm}
  \caption{V-SLAM}
  \label{VGlobal}
\end{subfigure}
\begin{subfigure}{.495\textwidth}
  \centering
   \includegraphics[width=\linewidth]{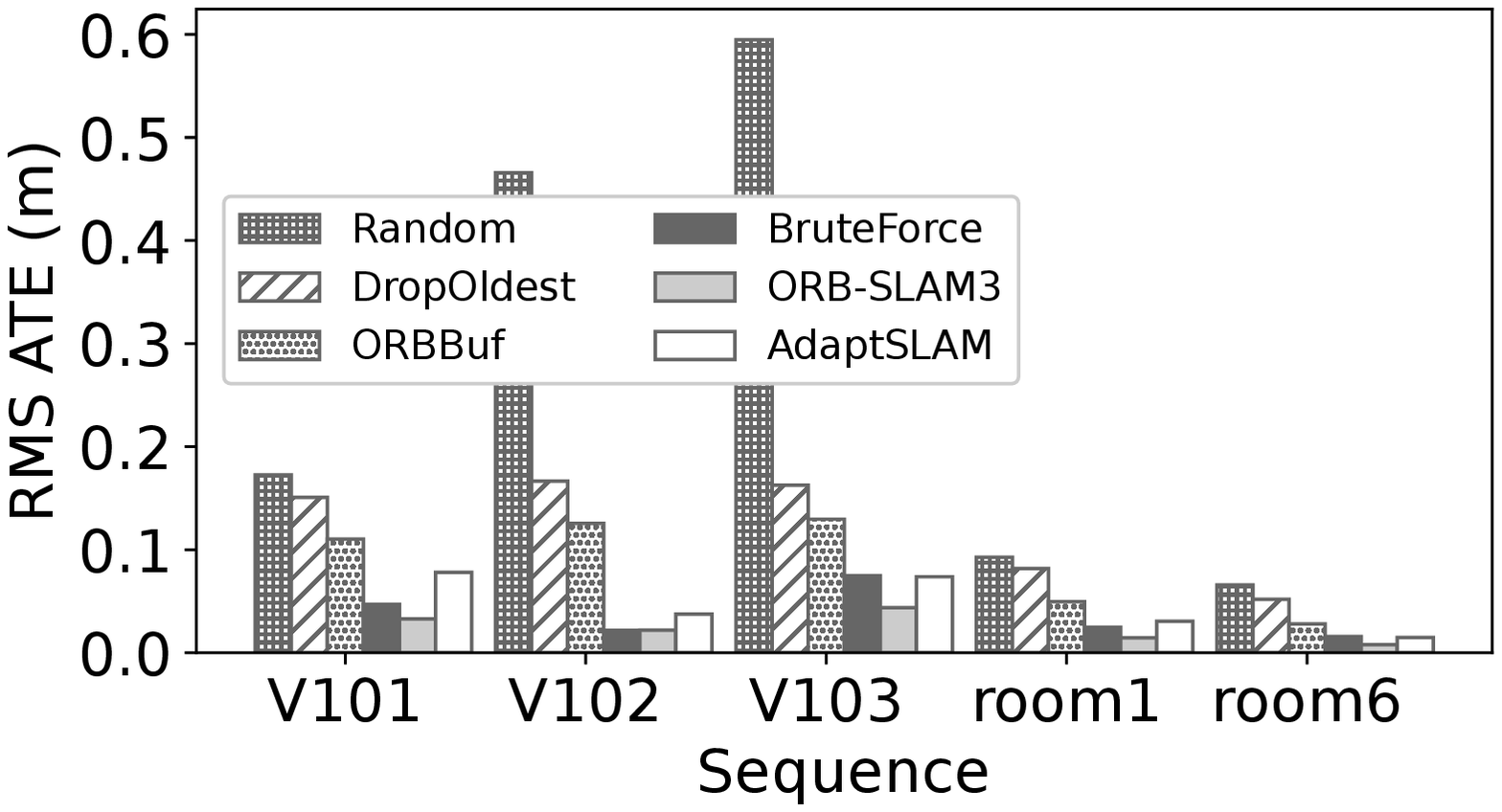}   
   \vspace{-0.6cm}
  \caption{VI-SLAM}
  \label{VIGlobal}
\end{subfigure}
\vspace{-0.15cm}
  \caption{RMS ATE for 6 keyframe selection methods in the global map construction for 5 sequences in EuRoC and TUM.
  }
     \label{Fig:Global}
\end{minipage}
\hspace{0.1cm}
\begin{minipage}[t]{0.29\linewidth}
\centering
   \includegraphics[width=0.9\textwidth]{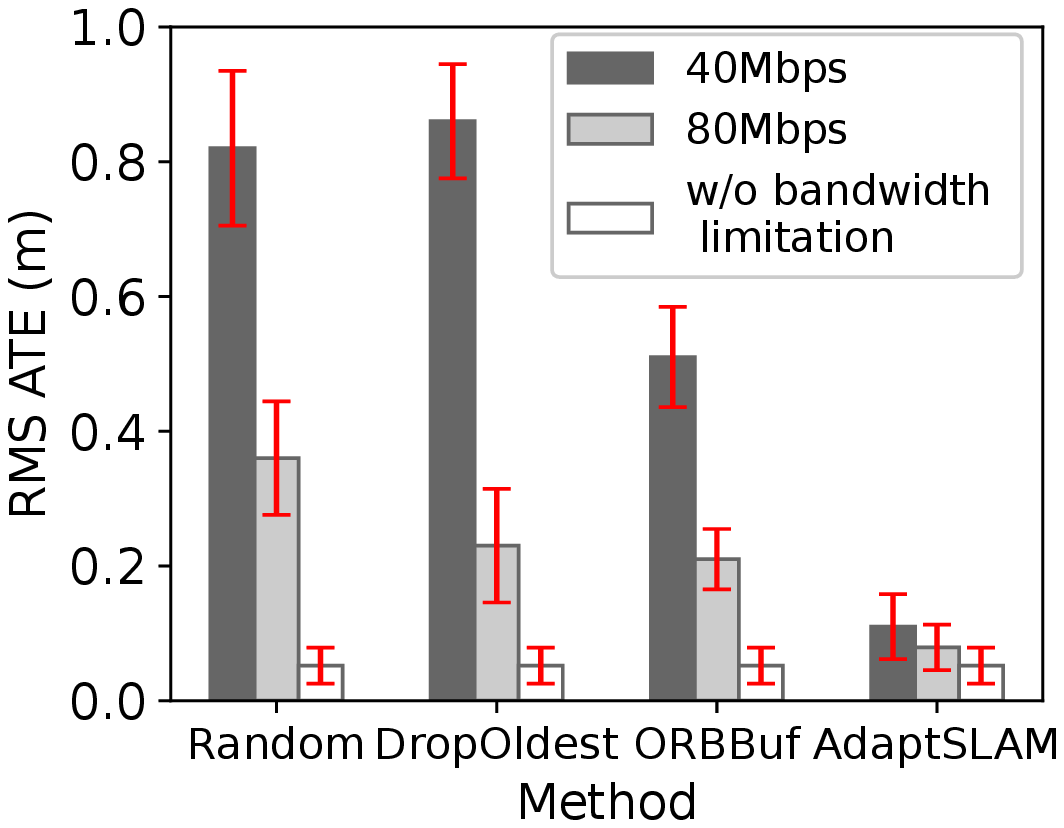}
   \vspace{-0.2cm}
  \caption{RMS ATE for different available bandwidth for offloading keyframes (for EuRoC V102).}
  \label{bandwidth}
\end{minipage}
\vspace{-0.7cm}
\end{figure*}

We implement AdaptSLAM on the open-source  ORB-SLAM3~\cite{orbslam3} framework which typically outperforms
older SLAM methods~\cite{klein2007parallel,dong2019pair}, with both V- and VI- configurations. The edge server modules are run on a Dell XPS 8930 desktop with Intel (R) Core (TM) i7-9700K CPU@3.6GHz and NVIDIA
GTX 1080 GPU under Ubuntu 18.04LTS. In \S\ref{sec:sim}, the mobile device
modules are run on the same desktop under simulated computation and network constraints. In \S\ref{sec:implementation}, the mobile device modules are implemented on a laptop
 (with an AMD Ryzen 7 4800H CPU and an NVIDIA
GTX 1660 Ti GPU), using a virtual machine with 4-core CPUs
and 8GB of RAM. 
The weight $w_e$, $e=((n,m),c)$ is set as the number of common map features visible in keyframes $n$ and $m$ for covisibility edges, similar to~\cite{orbslam3,Broadcast}, and  the  IMU edge weight is set as a large value (i.e., $500$) as the existence of IMU measurements greatly reduces the tracking error. We empirically set $H=5$ and $l_{thr}=30$ in Algorithm~\ref{alg_opt} to ensure low complexity and good performance at the same time.

\textbf{Metric.} 
We use root mean square (RMS) absolute trajectory error (ATE) as the SLAM performance metric which is commonly  used in the literature~\cite{orbslam3,zhang2018tutorial}. ATE is the absolute distance between the estimated and ground truth trajectories. 

\textbf{Baseline methods.}
We compare AdaptSLAM with 5 baselines. 
\textbf{Random} selects the keyframe randomly. \textbf{DropOldest} drops the oldest keyframes when the number of keyframes is constrained. 
\textbf{ORBBuf}, proposed in~\cite{ORBBuf}, chooses the keyframes that maximize the minimal edge weight between the adjacent selected keyframes. \textbf{BruteForce}  examines all the combinations of keyframes 
to search for the optimal one that minimizes the uncertainty (in Problems~\ref{prb1} and~\ref{prb2}). BruteForce can achieve better SLAM performance than AdaptSLAM but is shown to have exponential  computation complexity in \S\ref{sec:LocalProblem}.  In the \textbf{original ORB-SLAM3}, 
the local map includes all covisibility keyframes, and the global map includes all keyframes. The original ORB-SLAM3 also achieves better SLAM performance and consumes more computation resources than AdaptSLAM as the numbers of keyframes in both local and global maps are large.

\textbf{Datasets.} We evaluate 
AdaptSLAM on public SLAM datasets containing V and VI sequences, including  TUM~\cite{TUMdataset} and EuRoC~\cite{burri2016euroc}. The difficulty of a SLAM sequence depends on the extent of device mobility and scene illumination. 
We use EuRoC sequences V101 (easy), V102 (medium), and V103 (difficult), and difficult TUM VI room1 and  room6 sequences. 
We 
report the results over 10 trials for each sequence.

\subsection{Simulated Computation and Network Constraints}
\label{sec:sim}

First, we limit the number of keyframes in the local map under computation constraints, and all keyframes are used to build the global map without communication constraints. Second, we maintain local maps as in the default settings of ORB-SLAM3, and limit the number of keyframes in the global map under constrained communications, where $D$ in Problem~\ref{prb2} is set according to the available bandwidth.

\textbf{Local map construction}. 
We demonstrate the RMS ATE of different keyframe selection methods, 
for different V-SLAM 
(Fig.~\ref{VLocal}) and VI-SLAM (Fig.~\ref{VILocal}) sequences.   
 The size of the local map is limited to 10 keyframes and 9 
 anchors in V-SLAM sequences, and 25 keyframes and 10
 anchors in VI-SLAM sequences (to ensure successful tracking while keeping a small local map). 
AdaptSLAM reduces the RMS ATE compared with Random, DropOldest, and ORBBuf \emph{by more than 70\%, 62\%, and 42\%}, averaged over all sequences. 
The performance of AdaptSLAM is close to BruteForce, which demonstrates that our greedy-based algorithms yield near-optimal solutions, with substantially reduced computational complexity. 
Moreover, \emph{the performance of AdaptSLAM is 
close to the original ORB-SLAM3} (less than 0.05~m RMS ATE difference for all sequences) \emph{even though the size of the 
local map is reduced by more than 75\%}.

The influence of the number $l_{loc}$ of keyframes in the local map  on the RMS ATE for different methods is shown in Fig.~\ref{localsize}. We present the results for EuRoC V102 (of medium difficulty), which are representative. 
When $l_{loc}$ is reduced from 30 to 10,  AdaptSLAM increases the RMS ATE by only 6.7\%, to 0.09~m, as compared to 0.37, 0.16, and 0.12 m for, 
correspondingly, Random, DropOldest, and ORBBuf.  
This indicates that AdaptSLAM achieves low tracking error under stringent computation resource constraints.

\textbf{Global map construction}. 
First,  we examine the case where only half of all keyframes are offloaded to build a global map, for V-SLAM (Fig.~\ref{VGlobal}) and VI-SLAM (Fig.~\ref{VIGlobal}) sequences. 
AdaptSLAM reduces the RMS ATE compared with the closest baseline ORBBuf \emph{by 27\% and 46\% on average} for V- and VI-SLAM, and has small performance loss compared with the original ORB-SLAM3, despite reducing the number of keyframes by half. 

Next, in Fig.~\ref{bandwidth}, we examine four methods
whose performance is impacted by the available bandwidth, under different levels of communication constraints.
 Without bandwidth limitations, 
all methods have the same performance as the global map holds all keyframes.  
When the 
bandwidth is limited, 
Random and DropOldest have the worst performance 
as they 
ignore the relations of keyframes in the pose graph. The ORBBuf performs better, but the tracking error is increased by 4.0$\times$ and 9.8$\times$ when the  bandwidth is limited to 80 and 40~Mbps. AdaptSLAM achieves the best performance, \emph{reducing the RMS ATE compared to ORBBuf by 62\%  
and 78\% when network bandwidth is 80 and 40~Mbps, correspondingly}. 
This 
highlights the superiority of AdaptSLAM in achieving high tracking accuracy under communication constraints.

\begin{minipage}[t]{0.49\textwidth}
\vspace{-0.2cm}
\hspace{-0.6cm}
  \begin{minipage}[t]{0.43\textwidth}
    \centering
     \includegraphics[width=0.9\linewidth]{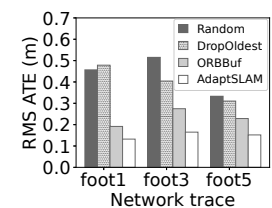}
    \captionof{figure}{RMS ATE for difference network traces. 
    }
    \label{Fig:networkTrace}
  \end{minipage}
  \hspace{-0cm}
  \begin{minipage}[t]{0.56\textwidth}
  \vspace{-2.9cm}
    \centering
\begin{tabular}{ cc } 
 \hline
 Method & Latency (ms)  \\
 \hline
 Random & 133.0$ \pm $86.3  \\
 \hline
 DropOldest & 139.9$ \pm $53.7\\
 \hline
ORBBuf & 149.3$ \pm $75.6 \\ 
 \hline
 BruteForce & 863.4$ \pm $123.5\\
 \hline
ORB-SLAM3 & 556.4$ \pm $113.7 \\ 
\hline
 AdaptSLAM & 162.8$ \pm $68.9\\
 \hline
\end{tabular}
\captionof{table}{The latency for local map construction and optimization. 
}
 \label{tab:latency}
    \end{minipage}
  \end{minipage}

\subsection{Real-World Computation and Network Constraints}
\label{sec:implementation}
            
Following the approach of splitting modules between the edge server and the mobile device~\cite{ben2020edge}, 
we split the modules as shown in Fig.~\ref{architecture}. The 
server and the 
device are connected via a network cable to minimize other factors. To ensure reproducibility, we replay the network traces  collected from a 
4G network~\cite{van2016http}. Focusing on mobile devices carried by users, we choose   network traces (foot1, foot3, and foot5)  collected by pedestrians. We set $D$ in Problem~\ref{prb2} according to the 
traces.

We examine the RMS ATE under the network traces in Fig.~\ref{Fig:networkTrace} for the EuRoC V102 sequence. 
The results 
for only four methods are presented because the overall time taken for running the SLAM modules onboard is high for BruteForce and the original SLAM. \emph{AdaptSLAM reduces the RMS ATE by 65\%, 61\%, and 35\%} (averaged over all 
traces) compared with Random, DropOldest, and ORBBuf.  
AdaptSLAM 
achieves high tracking accuracy under real-world network traces. 

Table~\ref{tab:latency} shows the computation latency of mobile devices for all six methods. We compare the latency for running local map construction and optimization, which is the main source of latency for modules running onboard~\cite{ben2020edge}. 
Compared with AdaptSLAM, 
the original ORB-SLAM3 takes 3.7$\times$ as much time for optimizing the local map as all covisibility keyframes are included in the local map without keyframe selection. Without the edge-assisted architecture, the original ORB-SLAM3 also runs global mapping and loop closing onboard which have even higher latency~\cite{ben2020edge}. BruteForce takes 5.3$\times$ as much time for examining all the combinations of keyframes to minimize the local map uncertainty. 
The latency for constructing and optimizing local maps using AdaptSLAM is  close to that using Random and DropOldest ($<$12.3\% difference). Low latency for local mapping shows that edge-assisted SLAM is appealing, as local mapping is the biggest source of delay for modules executing onboard after offloading
the intensive tasks (loop closing and global mapping).

\section{Conclusion}
\label{sec:conclusion}
We present AdaptSLAM, an edge-assisted SLAM that efficiently select subsets of keyframes to build local and global maps, under constrained communication and computation resources. AdaptSLAM quantifies the pose estimate uncertainty of V- and VI-SLAM under the edge-assisted architecture, and  minimizes the uncertainty by 
low-complexity algorithms based on the approximate submodularity properties and computation reuse. AdaptSLAM is demonstrated to 
reduce the size of the local keyframe set by 75\% compared with the original ORB-SLAM3 with a small performance loss.

\section*{Acknowledgments}
 This work was supported in part by NSF grants CSR-1903136, CNS-1908051, and CNS-2112562, NSF CAREER Award IIS-2046072, by an IBM Faculty Award, and by the Australian Research Council under Grant DP200101627.

\bibliographystyle{IEEEtran}
\bibliography{bibliography}

\appendix

\subsection{Proof of Lemma~\ref{lemma:global}}
\label{proof:global}
The quadratic term  of the objective function in~\eqref{globalopt} is $\sum\limits_{e = \left( {\left( {n,m} \right),c} \right) \in {{\cal E}_{glob}}} {{\bf{p}}_{n,m}^ \top } {{\cal I}_e}{{\bf{p}}_{n,m}}$, where ${{\bf{p}}_{n,m}^ \top } {{\cal I}_e}{{\bf{p}}_{n,m}}$ can be rewritten as
\begin{equation}
\begin{aligned}
&{{\bf{p}}_{n,m}^ \top } {{\cal I}_e}{{\bf{p}}_{n,m}}
\\=&w_e\left( {{{\bf{p}}_n^ \top},{{\bf{p}}_m^ \top}} \right)\left( {\begin{array}{*{20}{c}}
{  {\mathbf{I}_6}{\mathcal{I}}}&{-{\mathbf{I}_6}{\mathcal{I}  }}\\
-{{\mathbf{I}_6}{\mathcal{I} }}&{ {\mathbf{I}_6}{\mathcal{I} }}
\end{array}} \right)\left( {{{\bf{p}}_n},{{\bf{p}}_m}} \right)\\
 =& \mathbf{w}_g{\Xi }_e{\mathbf{w}_g^\top},
\end{aligned}
\end{equation}
where the $i,j$-th block of $\Xi_e$, $\left[\Xi_e\right]_{i,j}$, is derived as
\[{\left[ {\Xi _e} \right]_{i,j}} = \left\{ {\begin{array}{*{20}{c}}
{ - w_e{\cal I},}&{u_i = {n},u_j = {m}}\\
{w_e{\cal I},}&{u_i = u_j = n}\\
{0,}&{{\text{otherwise}}}
\end{array}} \right..\]
From the definition of $ \mathcal{I}_{glob}\left(\mathcal{K}_{g,edge}\right)$ 
and the global pose graph optimization formulation in \S\ref{sec:serverFunc}, we can obtain that $ \mathcal{I}_{glob}\left(\mathcal{K}_{g,edge}\right)={\sum\limits_{e \in {\mathcal{E}_{g,edge}}}  }\Xi_e$. Hence,  
$\mathbf{L}_{glob}$ is given by  \eqref{eq:Lglob},
which concludes the proof.

\subsection{Proof of Lemma~\ref{lemma:local}}
\label{proof:local}
As introduced in \S\ref{sec:userFunc}, the local pose graph optimization is to solve $\mathop {\min }\limits_{\{{\tilde{\bf{P}}}_{{n}}\}_{{n} \in {\mathcal{K}_{loc}}}} \sum\limits_{e \in {{\cal E}_{loc}} \cup {{\cal E}_{l,f}}} 
{{{\left( {{\bf{x}}_e} \right)}^\top}
\mathcal{I}_ e
{\bf{x}}_e}$. 
In optimizing  poses of keyframes in $\mathcal{K}_{loc}$, the  poses of keyframes in $\mathcal{K}_{fixed}$ are fixed. Hence, the quadratic term of the objective function can be rewritten as
\begin{equation*}
\label{eq:optnew}
\begin{aligned}
&\sum\limits_{e = \left( {\left( {n,m} \right),c} \right) \in {{\cal E}_{loc}} \cup {{\cal E}_{l,f}}} {{\bf{p}}_{n,m}^ \top } {{\cal I}_e}{{\bf{p}}_{n,m}} \\=& \sum\limits_{e = \left( {\left( {n,m} \right),c} \right) \in {{\cal E}_{loc}}} {{\bf{p}}_{n,m}^ \top } {{\cal I}_e}{{\bf{p}}_{n,m}} \\&+ \sum\limits_{e = \left( {\left( {n,m} \right),c} \right)\in {{\cal E}_{l,f}},n \in {{\cal K}_{loc}},m \in {{\cal K}_{fixed}} } {{\bf{p}}_n^ \top } {{\cal I}_e}{{\bf{p}}_n} \\&+ \sum\limits_{e = \left( {\left( {n,m} \right),c} \right)\in {{\cal E}_{l,f}},n \in {{\cal K}_{fixed}},m \in {{\cal K}_{loc}} } {\left( { - {\bf{p}}_m^ \top } \right)} {{\cal I}_e}\left( { - {{\bf{p}}_m}} \right).
\end{aligned}
\end{equation*}
 According to the above analysis, we can reformulate \eqref{localopt} as 
\begin{equation*}
\label{eq:proof2_2}
\mathop{\min}\limits_{\{\tilde{\mathbf{P}}_n\}_{n\in\mathcal{K}_{loc}}} \mathbf{w}_l\Lambda _{loc}\left(\mathcal{K}_{loc},\mathcal{K}_{fixed}\right) {\mathbf{w}_l^\top},
\end{equation*} 
where $\Lambda _{loc} \left(\mathcal{K}_{loc},\mathcal{K}_{fixed}\right)$ is the $\left|\mathcal{K}_{loc}\right|\times \left|\mathcal{K}_{loc}\right|$ block matrix whose $i,j$-th block is 
\[\begin{aligned}
&[\Lambda_{loc} \left(\mathcal{K}_{loc},\mathcal{K}_{fixed}\right)]_{i,j}=\\&\left\{ {\begin{array}{*{20}{c}}
{- \sum\limits_{e = \left( {( {{r_{{i}}},{r_{{j}}}} ),c} \right)\in {\mathcal{E}_{loc}}} {{w_e}} \mathcal{I},\;\;\;\;\;\;\;\;\;\;\;\;\;\;\;{i} \ne {j}}\\
{\begin{array}{l}
\left\{ {\sum\limits_{e = \left( {( {{r_{{i}}},q}),c} \right)\in {\mathcal{E}_{l,f}},q \in \mathcal{K}_{fixed}} {{w_e}} } \right. + \\
\left. {\sum\limits_{e = \left( {( {{r_{{i}}},q}),c} \right) \in {\mathcal{E}_{loc}},q \in \mathcal{K}_{loc},q \ne {r_{{i}}}} {{w_e}} } \right\}\mathcal{I}
\end{array},\ \ i = j}
\end{array}} \right.
\end{aligned}. \]

According to the uncertainty definition  in Definition~\ref{def:uncertainty}, the uncertainty of the local pose graph  is calculated as $-
\log \det \left( {\tilde{\cal I}_{loc}\left( {\mathcal{K}_{loc},\mathcal{K}_{fixed}} \right)}
\right)$, where ${\tilde{\mathcal{I}}_{loc}}\left(\mathcal{K}_{loc},\mathcal{K}_{fixed}\right)$ is given by \eqref{eq:Lloc}.

\subsection{Proof of Lemma~\ref{bound}}
\label{proof:bound}

According to the definition of the submodularity ratio given in~\eqref{eq:ratio}, the submodularity ratio $\gamma$  of the objective function in Problem~\ref{subproblem1} can be calculated as \eqref{eq:gamma}, where \eqref{gamma1} follows from the definition of the submodularity ratio. The denominator of \eqref{eq:gamma}, denoted as $\varsigma $, is lower bounded by

\begin{figure*}
\begin{equation}
\label{eq:gamma}
\begin{aligned}
\gamma& \labrel={gamma1}  \mathop {\min }\limits_{L \subseteq {{\cal K}_{add}},S \subseteq {{\cal K}_{add}},\left| S \right| \leqslant {l_{loc}} - {l_b},x \in {{\cal K}_{add}},x\not  \in S \cup L} \frac{{ - \mathsf{Unc}\left( {{{\cal K}_{base}} \cup L \cup \left\{ x \right\}} \right) + \mathsf{Unc}\left( {{{\cal K}_{base}} \cup L} \right)}}{{ - \mathsf{Unc}\left( {{{\cal K}_{base}} \cup L \cup S \cup \{ x\} } \right) + \mathsf{Unc}\left( {{{\cal K}_{base}} \cup L \cup S} \right)}}
\\&= \mathop {\min }\limits_{L \subseteq {{\cal K}_{add}},S \subseteq {{\cal K}_{add}},\left| S \right| \leqslant {l_{loc}} - {l_b},x \in {{\cal K}_{add}},x\not  \in S\cup L} \frac{{\log \frac{{\det \left( {{{\tilde {\cal I}}_{loc}}\left( {{{\cal K}_{base}} \cup L \cup \left\{ {x,k} \right\},\emptyset } \right)} \right)}}{{\det \left( {{{\tilde {\cal I}}_{loc}}\left( {{{\cal K}_{base}} \cup L \cup \left\{ k \right\}} \right),\emptyset } \right)}}}}{{\log \frac{{\det \left( {{{\tilde {\cal I}}_{loc}}\left( {{{\cal K}_{base}} \cup L \cup S \cup \left\{ {x,k} \right\},\emptyset } \right)} \right)}}{{\det \left( {{{\tilde {\cal I}}_{loc}}\left( {{{\cal K}_{base}} \cup L \cup S \cup \left\{ k \right\}} \right),\emptyset } \right)}}}}.
\end{aligned}
\end{equation}

\end{figure*}

\begin{figure*}
\begin{equation}
\begin{aligned}
\label{eq:gamma2}
&\gamma  = 1 + \frac{\log \left( {\mathop {\min }\limits_{L \subseteq {{\cal K}_{add}},S \subseteq {{\cal K}_{add}},\left| S \right| \leqslant {l_{loc}} - {l_b},x \in {{\cal K}_{add}},x\not  \in S\cup L} \frac{{\det \left( {{{\tilde {\cal I}}_{loc}}\left( {{{\cal K}_{base}} \cup L \cup \left\{ {x,k} \right\},\emptyset } \right)} \right)\det \left( {{{\tilde {\cal I}}_{loc}}\left( {{{\cal K}_{base}} \cup L \cup S \cup \left\{ k \right\}} \right),\emptyset } \right)}}{{\det \left( {{{\tilde {\cal I}}_{loc}}\left( {{{\cal K}_{base}} \cup L \cup S \cup \left\{ {x,k} \right\},\emptyset } \right)} \right)\det \left( {{{\tilde {\cal I}}_{loc}}\left( {{{\cal K}_{base}} \cup L \cup \left\{ k \right\}} \right),\emptyset } \right)}}} \right)} {\varsigma }
\\&=1 + \frac{\log \left( {\mathop {\min }\limits_{L \subseteq {{\cal K}_{add}},S \subseteq {{\cal K}_{add}},\left| S \right| \leqslant {l_{loc}} - {l_b},x \in {{\cal K}_{add}},x\not  \in S\cup L} \det \left( {\frac{{\left( {{\mathbf{Q}_1} + {\mathbf{Q}_2} + \left[ {\begin{array}{*{20}{c}}
{{{\bf{0}}_z}}&{}&{}\\
{}&{{{\bf{I}}_{\left| S \right|}}}&{}\\
{}&{}&0
\end{array}} \right]} \right)\left( {{\mathbf{Q}_1} + {\mathbf{Q}_3} + \left[ {\begin{array}{*{20}{c}}
{{{\bf{0}}_{z + \left| S \right|}}}&{}\\
{}&1
\end{array}} \right]} \right)}}{{\left( {{\mathbf{Q}_1} + {\mathbf{Q}_2} + {\mathbf{Q}_3} + {\mathbf{Q}_4}} \right)\left( {{\mathbf{Q}_1} + \left[ {\begin{array}{*{20}{c}}
0&0\\
0&{{{\bf{I}}_{{\left| S \right| + 1}}}}
\end{array}} \right]} \right)}}} \right)} \right)}{{\varsigma }}
\\ &= 1+ \frac{\log \left( {\mathop {\min }\limits_{L \subseteq {{\cal K}_{add}},S \subseteq {{\cal K}_{add}},\left| S \right| \leqslant {l_{loc}} - {l_b},x \in {{\cal K}_{add}},x\not  \in S\cup L} \det \frac{{{{\left( {{\mathbf{Q}_1}} \right)}^2} + {\mathbf{Q}_1}{\mathbf{Q}_2} + {\mathbf{Q}_1}{\mathbf{Q}_3} + \left( {{\mathbf{Q}_2} + {\mathbf{Q}_3}} \right)\left[ {\begin{array}{*{20}{c}}
0&0\\
0&{{{\bf{I}}_{\left| S \right| + 1}}}
\end{array}} \right] + {\mathbf{Q}_2}{\mathbf{Q}_3}}}{{{{\left( {{\mathbf{Q}_1}} \right)}^2} + {\mathbf{Q}_1}{\mathbf{Q}_2} + {\mathbf{Q}_1}{\mathbf{Q}_3} + \left( {{\mathbf{Q}_2} + {\mathbf{Q}_3}} \right)\left[ {\begin{array}{*{20}{c}}
0&0\\
0&{{{\bf{I}}_{\left| S \right| + 1}}}
\end{array}} \right] + {\mathbf{Q}_4}}}} \right)}{\varsigma }
\\& \labrel\geqslant{gamma2} 1 + \frac{1}{\vartheta }\log \left( {1 - {\mathbf{g}_1}{{ {\mathbf{Q}} }^{ - 1}}\mathbf{g}_1^\top} \right) .
\end{aligned}
\end{equation}
\hrulefill
\end{figure*}

\begin{equation}
\begin{aligned}
\label{eq:dominator}
\varsigma =&\log\frac{{\det \left( {{{\tilde {\cal I}}_{loc}}\left( {{{\cal K}_{base}} \cup L \cup S \cup \left\{ {x,k} \right\},\emptyset } \right)} \right)}}{{\det \left( {{{\tilde {\cal I}}_{loc}}\left( {{{\cal K}_{base}} \cup L \cup S \cup \left\{ k \right\}} \right),\emptyset } \right)}}
\\ \geqslant &  \sum\limits_{n \in  {{{\cal K}_{base}} \cup L \cup S} } {\log {w_{x,{n}}}} 
 \\
\geqslant &\mathop {\min }\limits_{{m} \in  {{{\cal K}_{add}}} } \sum\limits_{{n} \in  {{{\cal K}_{base}}} } {\log {w_{{n},{m}}}}  \triangleq \vartheta,
\end{aligned}
\end{equation}
where the first inequality is due to the fact that the determinant of the reduced weighted Laplacian matrix  is equal to the tree-connectivity of its corresponding graph~\cite{khosoussi2019reliable}.

Substituting~\eqref{eq:dominator} into~\eqref{eq:gamma}, $\gamma$ can be further calculated by~\eqref{eq:gamma2}, where $\mathbf{I}_i$ and $\mathbf{0}_i$ are the $i \times i$ identity matrix and zero matrix, 
and $\mathbf{Q}=\mathbf{Q}_1+\mathbf{Q}_2+\mathbf{Q}_3+\mathbf{Q}_4$. $\mathbf{Q}_1$, $\mathbf{Q}_2$, $\mathbf{Q}_3$ and $\mathbf{Q}_4$ are defined as follows. We express ${{\cal K}_{base}} \cup L \cup S \cup \{  x, k\}$ as ${\{ {s_i}\} _{i = \{ 1, \cdots ,z+|S|+ 2\} }}$ where $z=|\mathcal{K}_{base}\cup L|$, $s_i \in \mathcal{K}_{base}\cup L$ when $i\leqslant  z$, $s_i \in S$ when $z<i\leqslant z+|S|$, $s_{z+|S|+1}=x$, and $s_{z+|S|+2}=k$. For each edge $e$, we define a vector $\mathbf{q}_e$, and each element $[\mathbf{q}_e]_{i}=-[\mathbf{q}_e]_{j}=w_e$ if vertexes $s_i$ and $s_j$ are the head or tail of $e$ and zero otherwise. We then get $\tilde{\mathbf{q}}_e$ after removing the last element of $\mathbf{q}_e$. $\mathbf{Q}_1$, $\mathbf{Q}_2$, $\mathbf{Q}_3$ and $\mathbf{Q}_4$ are defined as ${{\bf{Q}}_1} = \frac{1}{2}\sum\limits_{e = \left( {( {{s_i},{s_j}}),c} \right),{s_i},{s_j} \in {{\cal K}_{base}} \cup L} {{{\widetilde {\bf{q}}}_e}} {{{\widetilde {\bf{q}}}_e^\top}}$ ($\frac{1}{2}$ is used because the edges from $s_i$ to $s_j$ and from $s_j$ to $s_i$ are both included), ${{\bf{Q}}_2} = \sum\limits_{e = \left( {( {{s_i},x} ),c} \right),{s_i} \in {{\cal K}_{base}} \cup L} {{{\widetilde {\bf{q}}}_e}}{{{\widetilde {\bf{q}}}_e^\top}}$, ${{\bf{Q}}_3} = \sum\limits_{e = \left( {( {{s_i},{s_j}} ),c} \right),{s_i} \in  {{{\cal K}_{base}} \cup L},{s_j} \in S} {{{\widetilde {\bf{q}}}_e}}{{{\widetilde {\bf{q}}}_e^\top}}$, and ${{\bf{Q}}_4} = \sum\limits_{e = \left( {( {x,{s_j}} ),c} \right),{s_j} \in S} {{{\widetilde {\bf{q}}}_e}} {{{\widetilde {\bf{q}}}_e^\top}}$.
$\mathbf{g}_1$  
is given by
\[{\mathbf{g}_1} = \left[ {\underbrace {0, \cdots ,0}_{z{\rm{ }}\,`{\rm{0\rq s}}},-{w_{x,{s_1}}}, \cdots ,-{w_{x,{s_{\left| S \right|}}}}\sum\limits_{i = 1}^{|S|} {{w_{x,{s_i}}}} } \right],\]
where ${w_{\max }} = \mathop {\max }\limits_{{n},{m} \in  {{{\cal K}_{base}} \cup {{\cal K}_{add}}} } {w_{{n},{m}}}$, and $w_{x,s_{i}}\leqslant w_{max}$ for $s_i \in S$.
\eqref{gamma2} in~\eqref{eq:gamma2} is because for invertible positive semidefinite matrices $\mathbf{M}$, $\mathbf{N}$,  $\det (\mathbf{M})\geqslant\det(\mathbf{N})$ holds when $\mathbf{M}-\mathbf{N}$ is positive semidefinite~\cite{strang2006linear}.

We will prove that $\left\| {{{\bf{Q}}^{ - 1}}} \right\| \leqslant  \frac{1}{{\left| {{\mathcal{K}_{base}}} \right|{w_{\min }} - {w_{\max }}}}$ when $\left| {{{\cal K}_{base}}} \right|$ is significantly larger than $ \left| {{{\cal K}_{add}}} \right|$, where $\|\mathbf{M}\|$ is the $l_\infty $ norm of $\mathbf{M}$ (defined as the largest magnitude among each element in $\mathbf{M}$), and ${w_{\min }} = \mathop {\min }\limits_{{n} ,{m} \in  {{{\cal K}_{base}} \cup {{\cal K}_{add}}} } {w_{{n},{m}}}$. Rewrite ${\bf{Q}}$ as ${\bf{Q}} = {{\bf{D}}_Q} - {{\bf{E}}_Q}$, where $\mathbf{D}_Q$ is a diagonal matrix with elements on the diagonal  the same as those of  $\bf{Q}$, and $\mathbf{E}_Q=\mathbf{D}_Q-\mathbf{Q}$. ${\bf{Q}}^{-1}$ is calculated as 
\begin{equation*}
\begin{aligned}&{{\bf{Q}}^{ - 1}} = {\left( {{{\bf{I}}_{z + \left| S \right| + 1}} - {\bf{D}}_Q^{ - 1}{{\bf{E}}_Q}} \right)^{ - 1}} {\bf{D}}_Q^{ - 1}\\= &\left[\sum\limits_{i = 0}^\infty  {{{\left( {{\bf{D}}_Q^{ - 1}{{\bf{E}}_Q}} \right)}^i}}\right]
{\bf{D}}_Q^{ - 1}. 
\end{aligned}
\end{equation*}
${\bf{D}}_Q^{ - 1}{{\bf{E}}_Q}$ has the properties that all elements in ${\bf{D}}_Q^{ - 1}{{\bf{E}}_Q}$ are positive and smaller than $\frac{{{w_{\max }}}}{{\left| {{\mathcal{K}_{base}}} \right|{w_{\min }}}}$, and all row vectors has an $l_\infty$ norm smaller than 1. Hence, we have ${\left\| \left({{\bf{D}}_Q^{ - 1}{{\bf{E}}_Q}} \right)^i\right\|} \leqslant \frac{{{w_{\max }}}}{{\left| {{{\cal K}_{base}}} \right|{w_{\min }}}}$.

${{{\bf{Q}}^{ - 1}}} $ is given by
\begin{equation}
\label{eq:Qinverse}
\begin{aligned}
&\left\| {{{\bf{Q}}^{ - 1}}} \right\| \leqslant \frac{1}{{\left| {{{\cal K}_{base}}} \right|{w_{\min }}}}\left\| {\sum\limits_{i = 1}^\infty  {{{\left( {{\bf{D}}_Q^{ - 1}{{\bf{E}}_Q}} \right)}^i}} } \right\|
\\&\leqslant \frac{1}{{\left| {{{\cal K}_{base}}} \right|{w_{\min }}}}\frac{1}{{1 - \left\| {{\bf{D}}_Q^{ - 1}{{\bf{E}}_Q}} \right\|}} 
\\&\leqslant \frac{1}{{\left| {{\mathcal{K}_{base}}} \right|{w_{\min }}}}\frac{1}{{1 - \frac{{{w_{\max }}}}{{\left| {{\mathcal{K}_{base}}} \right|{w_{\min }}}}}}
\\&= \frac{1}{{\left| {{\mathcal{K}_{base}}} \right|{w_{\min }} - {w_{\max }}}}.
\end{aligned}
\end{equation}

	Substituting~\eqref{eq:Qinverse} into~\eqref{eq:gamma2}, we can derive that $\gamma \geqslant 1 + \frac{1}{\vartheta }\log \left( {1 - \frac{{4{{\left| {{{\cal K}_{add}}} \right|}^2}{w^2_{{{\max }}}}}}{{\left| {{{\cal K}_{base}}} \right|{w_{\min }} - {w_{\max }}}}} \right)$.

\end{document}